\documentclass[11pt]{article}

\usepackage[margin=1in]{geometry}
\usepackage[T1]{fontenc}
\usepackage[utf8]{inputenc}
\usepackage{lmodern}
\usepackage{microtype}
\usepackage{amsmath,amssymb,amsthm,mathtools,bm}
\usepackage{booktabs}
\usepackage{array}
\usepackage{tabularx}
\usepackage{float}
\usepackage{enumitem}
\usepackage{siunitx}
\usepackage{xcolor}
\usepackage{graphicx}
\usepackage{pgfplots}
\usepackage{hyperref}
\usepackage[nameinlink,capitalize]{cleveref}

\graphicspath{{figures/}}
\pgfplotsset{compat=1.18}
\usepgfplotslibrary{groupplots}
\usepgfplotslibrary{colormaps}
\usetikzlibrary{arrows.meta,positioning}

\definecolor{euvAcol}{HTML}{1B4965}
\definecolor{euvBcol}{HTML}{BC4749}
\definecolor{euvCcol}{HTML}{2A9D8F}
\pgfplotsset{
  euvA/.style={color=euvAcol},
  euvB/.style={color=euvBcol},
  euvC/.style={color=euvCcol},
  euvaxis/.style={
    tick label style={font=\footnotesize},
    label style={font=\small},
    title style={font=\small},
    legend style={font=\footnotesize, draw=none, fill=none},
    grid=both,
    grid style={line width=.2pt, draw=gray!25},
    major grid style={line width=.3pt, draw=gray!35},
    axis line style={line width=.5pt},
    tick style={line width=.4pt},
  },
}

\hypersetup{colorlinks=true,linkcolor=blue!50!black,citecolor=blue!50!black,urlcolor=blue!55!black}
\setlist{leftmargin=*}
\newtheorem{theorem}{Theorem}
\newtheorem{proposition}{Proposition}

\newtheorem{corollary}{Corollary}
\newtheorem{definition}{Definition}
\newtheorem{experiment}{Experiment}
\newtheorem{assumption}{Assumption}
\newtheorem{remark}{Remark}

\DeclareMathOperator{\KL}{D_{\rm KL}}
\DeclareMathOperator{\TV}{TV}
\DeclareMathOperator{\Tr}{Tr}
\DeclareMathOperator{\Null}{Null}

\newcommand{\E}{\mathbb{E}}

\newcommand{\dd}{\mathrm{d}}
\newcommand{\eps}{\varepsilon}
\newcommand{\calA}{\mathcal A}
\newcommand{\calB}{\mathcal B}

\newcommand{\calI}{\mathcal I}
\newcommand{\calM}{\mathcal M}
\newcommand{\calN}{\mathcal N}
\newcommand{\calP}{\mathcal P}
\newcommand{\calU}{\mathcal U}
\newcommand{\calZ}{\mathcal Z}
\newcommand{\inner}[2]{\left\langle #1,#2\right\rangle}
\newcommand{\norm}[1]{\left\lVert #1\right\rVert}
\newcommand{\Adv}{\mathsf{Adv}}

\title{Quantifying Side-Channel Leakage in Public Metrology Releases\\[1.5ex]{\large\itshape Screened EUV Roughness Spectra as a Case Study}}
\author{
  Faruk Alpay\thanks{Correspondence: \texttt{alpay@lightcap.ai}}\\
  Department of Computer Engineering\\
  Bah\c{c}e\c{s}ehir University, Istanbul, Turkiye\\
  \texttt{faruk.alpay@bahcesehir.edu.tr}
  \and
  Taylan Alpay\\
  Department of Aerospace\\
  University of Turkish Aeronautical Association, Ankara, Turkiye\\
  \texttt{s220112602@stu.thk.edu.tr}
}
\date{May 31, 2026}

\begin{document}
\maketitle

\begin{abstract}
Public scientific and metrology releases can leak the hidden settings that produced them. We formalize and quantify this risk as a profiled statistical side-channel audit: a release map exposes finite-band statistics of a power spectral density (PSD), a profiled observer trains labeled template spectra under an explicit budget, and a challenge release is drawn from one of two utility-equivalent recipes separated by a protected coordinate. Averaged PSD bins follow a gamma channel, replaced by a covariance-weighted log-spectrum channel when bins are correlated, which yields exact KL divergences, Chernoff exponents, protected-bit advantage bounds, and finite-training, finite-library, finite-compute, and mismatch corrections. Our headline result is a finite-band transport-leakage law: after amplitude and blur are eliminated, the protected acid-transport information obeys \(\mathcal I_{\lambda\mid\alpha,\beta}(K)=\tfrac{64}{1225}w\lambda^6K^9+O(w\lambda^8K^{11})\) for \(K\lambda\ll1\), a ninth-order exponent with a closed-form safe band. A step-by-step protocol turns a measured release into these numbers, and a fixed-seed package reproduces every table and figure. We instantiate the audit on screened EUV roughness spectra as a model-conditioned case study, with deployment on measured releases the next step.
\end{abstract}

\section{Introduction}

Public scientific and metrology releases can leak the hidden settings that produced them: once a measurement is published, an observer may use it to recover process parameters the releaser never intended to disclose. This paper formalizes that risk as a profiled statistical side-channel audit. The hidden state is a recipe coordinate such as quencher loading, acid-transport length, stochastic amplitude, secondary-electron blur, or metrology floor; the public transcript is a released roughness statistic such as a full power-spectral-density (PSD) band, a redacted band, a fitted parameter vector, or an RMS-only scalar; and the audit asks how many released spectra are needed before a protected recipe bit becomes distinguishable within a public utility class. We adapt quantitative information-flow and profiled template-attack reasoning to public scientific releases, quantifying leakage as statistical-channel distinguishability under a declared release channel. The contribution is a reproducible statistical side-channel audit framework for scientific metrology releases, in which each adversary class is graded by a profiling budget: the labeled spectra it trains on, the templates it stores, and the likelihood scores it evaluates.

Our running case study is EUV lithography metrology, where roughness spectra are process measurements and release artifacts. RMS line-edge roughness compresses an edge into one scalar, while a PSD release preserves the low-frequency plateau, correlation length, high-frequency roll-off, and metrology floor. Model-function roughness analysis made this spectral representation standard for lithographic roughness characterization \cite{constantoudis2004}; the imec roughness protocol sharpened the requirements for reproducible PSD metrology \cite{lorusso2018}; and photoresist studies use PSD shape to compare formulation and process changes \cite{cutler2021}. The screened EUV spectrum below is a physically motivated instantiation of the audit, not a validated process model.

The criterion is statistical-channel distinguishability: it quantifies how distinguishable physical recipe bits are under the released PSD laws of a declared channel model. The adversary is operationally a profiled observer who can characterize the released channel on \emph{known} recipes (for instance a party that runs the same public metrology on reference wafers, or that has accumulated earlier releases whose process settings are known) and then attributes a challenge release by likelihood. Its baseline is the best likelihood test allowed by the released channel, refined by additional rows for finite training, finite template libraries, finite score evaluations, adaptive sublibrary search, and model mismatch.

The EUV core is the screened spectrum
\[
S(k)=A e^{-\beta k^2}(1+\lambda^2k^2)^{-3/2}+S_0,
\]
with \(\lambda^2=D_H/(k_qq_0+k_{\rm loss})\). The nuisance coordinates \(A\) and \(\beta\) absorb stochastic amplitude and Gaussian blur, and \(S_0\) is the metrology floor. Here \(\lambda\) is an \emph{effective} transport coordinate: a single released spectrum identifies only the combination \(\mu=\lambda^{-2}=(k_qq_0+k_{\rm loss})/D_H\), so attributing a shift to quencher loading \(q_0\) rather than to the loss rate requires a controlled sweep (\cref{sec:transport}); the leakage statements below concern \(\lambda\) (equivalently \(\mu\)), not an individually identified chemical rate. After amplitude and Gaussian blur are eliminated, the first acid-transport signature that survives in a low band is a fourth-order curvature residual whose squared finite-band norm scales as \(\lambda^6K^9\), producing a ninth-order leakage exponent below the transport knee.

\paragraph{Contributions.}
\begin{enumerate}[label=(\arabic*)]
\item A resource-explicit recipe-release game with training oracle, challenge transcript, protected bit, utility-preserving pairs, and budgets \((N,Q,L,T)\) that grade the adversary by profiling effort: training spectra, stored templates, and likelihood-score evaluations.
\item Exact gamma-channel likelihoods for averaged PSD release, and a covariance-weighted replacement for correlated log-PSD bins, with KL, Chernoff, protected-bit advantage, finite-training, finite-compute, and additive-accounting bounds.
\item A finite-band EUV transport theorem: \(\mathcal I_{\lambda\mid\alpha,\beta}(K)=(64/1225)w\lambda^6K^9+O(w\lambda^8K^{11})\).
\item A nuisance-coupled release analysis for nonzero floor, chemical non-identifiability, Schur-complement coupling, and optimal zero-leakage utility projection.
\item A step-by-step audit protocol (\cref{sec:protocol}) that turns a measured PSD release into the reported exponents and safe band, fixing the bin-variance estimator, the covariance estimate, the screened-model fit and its residual diagnostics, the nuisance profiling, and the mismatch and floor checks.
\item A numerical audit with a reconstructed-PSD fit (used only to exercise the pipeline, not measured data), a published 18-nm calibration, 10{,}000-trial simulations, PSD-vs-RMS ablation, finite-library thresholds, and floor-mismatch sensitivity.
\item An open, reproducible package (\cref{sec:validation}) that regenerates every table and figure with fixed seeds, numerically checks the ninth-order asymptotic law under the stated screened-PSD model, and adds finite-training, correlated-bin, coupling-map, and release-optimizer computations; the worked derivations are collected in \crefrange{app:green}{app:optimizers}.
\end{enumerate}

\paragraph{Roadmap.}
We proceed in five movements. \emph{First}, we build the release channel: we cast the screened EUV spectrum as a parametric family and turn each averaged PSD band into a gamma likelihood, from which we will read exact KL and Chernoff exponents (\cref{sec:channel}). \emph{Second}, we make the adversary explicit: we define the recipe-release game, fix the profiling budgets, and prove the optimal-test, protected-bit, and finite-resource bounds that the audit will report row by row (\cref{sec:game}). \emph{Third}, we isolate the physics: we eliminate amplitude and Gaussian blur as nuisance directions and show that the surviving acid-transport signal obeys the ninth-order law $\calI_{\lambda\mid\alpha,\beta}(K)=(64/1225)w\lambda^6K^9+O(K^{11})$, from which a closed-form safe band follows (\cref{sec:transport}). \emph{Fourth}, we turn the law into a design tool: we place protected tangents in the nullspace of the release map and solve for the public statistic that keeps the most utility under a leakage budget (\cref{sec:release}). \emph{Finally}, we give a step-by-step protocol for auditing a measured release (\cref{sec:protocol}) and let the code referee the theory: we reproduce the manuscript's calibrated and synthetic audits, then add transport-knee, finite-training, correlated-bin, coupling, and release-optimizer experiments that each stress a separate assumption (\crefrange{sec:audits}{sec:validation}), before concluding with the scope the guarantees do and do not cover. Every derivation invoked along the way is carried out in full in \crefrange{app:green}{app:optimizers}.

\paragraph{Reproducibility.} Every number below is produced by the ancillary package: \texttt{python scripts/reproduce\_all.py} regenerates the tables and figures, and \texttt{anc/reproduce.py} regenerates the core tables from inside the ancillary tree. The seed is \(20260531\) throughout. We label \emph{calibrated} (public 18-nm numbers), \emph{synthetic} (generated spectra), and \emph{reconstructed} (the digitized-PSD points, a stand-in for an unavailable original digitization) quantities explicitly. The layout is flat: \texttt{src/euv\_audit} holds the core channel, divergence, projection, and release-map code; \texttt{scripts/} holds the audit drivers, figure-data export, and the source packager; \texttt{tests/} holds unit and regression checks of the identities; \texttt{validation/} holds the Monte-Carlo evidence runs with frozen configs and per-run provenance; \texttt{paper/figures} and \texttt{paper/tables} hold the pgfplots \texttt{.tex} sources with their \texttt{.dat} data; and \texttt{anc/} is the standalone ancillary tree shipped with the source. The package targets Python\,\(\ge\)\,3.10 with NumPy, SciPy, and Matplotlib (pinned in \texttt{requirements.txt} and \texttt{environment.yml}), and a top-level \texttt{README} gives the one-command reproduce. A suite of \(51\) unit and regression checks across nine files pins the identities used here and runs in seconds: the gamma KL/Chernoff formulas, gamma moments, the Schur-complement transport information by three independent methods (QR, SVD, and direct), the ninth-order law on a small grid, covariance regularization, and fixed-seed determinism. The full Monte-Carlo validation completes in minutes.

\section{EUV spectral release channel}
\label{sec:channel}

We begin by fixing what is released. In this section we write down the recipe-resolved spectrum, identify which coordinates are protected and which are nuisances, and convert an averaged PSD band into a likelihood so that ``how much does a release leak'' becomes a divergence we can compute.

\subsection{Recipe-resolved spectrum}

A local recipe is \(\theta=(A,\beta,\lambda,S_0,\vartheta)\). The screened spectrum is
\begin{equation}
S_\theta(k)=A\exp(-\beta k^2)(1+\lambda^2 k^2)^{-3/2}+S_0,\qquad 0\leq k\leq K,
\label{eq:screened-psd}
\end{equation}
with transport length
\begin{equation}
\lambda^2=\frac{D_H}{k_q q_0+k_{\rm loss}}.
\label{eq:lambda-chem}
\end{equation}
The screened factor \((1+\lambda^2k^2)^{-3/2}\) is the one-dimensional edge spectrum of a two-dimensional screened reaction--diffusion field; the full reduction is carried out in \cref{app:green}. Coupled acid catalysis and diffusion were measured in chemically amplified resists \cite{houle2000}; acid-base reaction modeling motivates the screened loss term \cite{houle2004}; acid diffusion is a process-limiting coordinate \cite{vansteenwinckel2005}.

\subsection{Metrology-floor regimes}

With \(P(k)=A\exp(-\beta k^2)(1+\lambda^2k^2)^{-3/2}\), \(r(k)=P/(P+S_0)\), \(\alpha=\log A\), and \(f=\log S_\theta\), the floor-aware tangents are
\begin{equation}
\partial_\alpha f=r,\quad \partial_\beta f=-rk^2,\quad
\partial_\lambda f=-r\frac{3\lambda k^2}{1+\lambda^2k^2},\quad
\partial_{S_0}f=\frac{1}{P+S_0}.
\label{eq:floor-tangents}
\end{equation}

\begin{proposition}[Floor-coupled transport information]
\label{prop:floor-coupling}
Let \(\calN_r=\operatorname{span}\{\partial_\alpha f,\partial_\beta f\}\) and \(\calN_F=\operatorname{span}\{\partial_\alpha f,\partial_\beta f,\partial_{S_0}f\}\). The fixed-floor and profiled-floor conditional informations are
\[
\calI_{\lambda\mid\alpha,\beta}^{(r)}(K)=\inner{g_\lambda^F}{(I-\Pi_{\calN_r})g_\lambda^F}_{w,K},
\qquad
\calI_{\lambda\mid\alpha,\beta,S_0}(K)=\inner{g_\lambda^F}{(I-\Pi_{\calN_F})g_\lambda^F}_{w,K},
\]
with \(g_\lambda^F=-r\,3\lambda k^2/(1+\lambda^2k^2)\). If \(\rho_K=\sup_k S_0/P(k)\le\tau<1\), then \(|\calI_{\lambda\mid\alpha,\beta}^{(r)}-\calI_{\lambda\mid\alpha,\beta}|\le C_K\tau\norm{g_\lambda}_{w,K}^2\). Profiling the floor can only reduce information, with the exact one-direction loss obtained by projecting \(\partial_{S_0}f\) after \(\calN_r\). If \(r(k)\le r_*\) on a subband, that subband contributes at most \(r_*^2\) of the corresponding floor-free raw Fisher entry.
\end{proposition}

\begin{proof}
Work in the Hilbert space \(H=L^2_w[0,K]\). For any nuisance span \(\calN\), the Fisher Schur complement of \(g_\lambda^F\) after eliminating \(\calN\) is
\[
\inf_{h\in\calN}\|g_\lambda^F-h\|_{w,K}^2
=\|(I-\Pi_\calN)g_\lambda^F\|_{w,K}^2,
\]
because the normal equations for the least-squares projection are exactly the Schur-complement normal equations. Differentiating \(f=\log(P+S_0)\) gives the four tangents in \eqref{eq:floor-tangents}, hence the two displayed residual-norm formulas.

For the perturbation estimate, write \(M_r\) for multiplication by \(r=(1+\rho)^{-1}\). Since \(\rho_K\le\tau\), \(\|M_r-I\|_{H\to H}\le \tau\). The floor-fixed nuisance space is \(M_r\operatorname{span}\{1,-k^2\}\), while \(g_\lambda^F=M_rg_\lambda\). The two-dimensional Gram matrix of \(\{1,-k^2\}\) on \([0,K]\) is nonsingular, so the orthogonal projection onto this finite-dimensional span is locally Lipschitz under the bounded perturbation \(M_r=I+O(\tau)\). Thus \(\|\Pi_{\calN_r}-\Pi_{\alpha,\beta}\|\le c_K\tau\) for a constant depending only on the Gram conditioning over the released band. Expanding the difference of the two residual norms,
\[
\|(I-\Pi_{\calN_r})M_rg_\lambda\|^2-\|(I-\Pi_{\alpha,\beta})g_\lambda\|^2,
\]
and using \(\|M_r-I\|\le\tau\) and \(\|\Pi_{\calN_r}-\Pi_{\alpha,\beta}\|\le c_K\tau\) gives the stated bound with \(C_K=2(1+c_K)+O(\tau)\).

Finally, \(\calN_r\subseteq\calN_F\), so projection onto \(\calN_F\) cannot increase the residual norm. More explicitly, if \(h=(I-\Pi_{\calN_r})\partial_{S_0}f\), then
\[
\calI_{\lambda\mid\alpha,\beta}^{(r)}-\calI_{\lambda\mid\alpha,\beta,S_0}
=
\begin{cases}
|\inner{(I-\Pi_{\calN_r})g_\lambda^F}{h}_{w,K}|^2/\|h\|_{w,K}^2, & h\ne0,\\
0, & h=0.
\end{cases}
\]
On any subband where \(r\le r_*\), every non-floor recipe tangent is multiplied by \(r\), so every raw Fisher integrand is multiplied by \(r^2\le r_*^2\).
\end{proof}

\subsection{PSD estimator assumptions}

\begin{assumption}[PSD channel conditions]
\label{ass:psd}
All divergence, exponent, and finite-resource results below are stated under the following standing conditions. \textbf{(A1)} Bins are independent with \(\widehat S_i\mid\theta\sim{\rm Gamma}(m_i,S_\theta(k_i)/m_i)\) (mean \(S_\theta(k_i)\), variance \(S_\theta(k_i)^2/m_i\)). \textbf{(A2)} \(S_\theta>0\) and \(\theta\mapsto\log S_\theta\) is twice continuously differentiable on the parameter set. \textbf{(A3)} After quotienting the declared nuisance directions, the released-tangent Fisher form is nonsingular (the protected coordinate is identifiable). \textbf{(A4)} Template estimators are \(\sqrt{Q_j}\)-consistent and asymptotically efficient. \textbf{(A5)} The per-release model mismatch \(\KL(P_\theta\Vert Q_\theta)\) is uniformly bounded. Each condition is empirically checkable on a measured release, as detailed next and assembled into a procedure in \cref{sec:protocol}.
\end{assumption}

A binwise estimate of the effective shape is \(\widehat m_i=(\bar S_i/s_i)^2\) from repeated spectra, and Welch averaging is the classical setting \cite{welch1967}. These conditions are checkable on a measured release rather than merely assumed: the gamma shape A1 by a per-bin goodness-of-fit test on the repeats, the effective degrees of freedom by \(2\widehat m_i\), the \(C^2\) regularity A2 by inspecting log-PSD residuals for unmodeled structure, and the nonsingularity A3 from the conditioning of the profiled nuisance block. \cref{sec:protocol} collects these diagnostics into an explicit procedure.

The gamma shape A1 is the standard model for a Welch-averaged periodogram: averaging \(K_{\rm seg}\) \emph{non-overlapping} segments makes each bin approximately \(\mathrm{Gamma}(m_i,S_\theta(k_i)/m_i)\) with \(m_i\approx K_{\rm seg}\) and \(2m_i\) degrees of freedom, and the bins are then approximately independent. Overlapped segments, the usual variance-reduction choice, instead correlate neighbouring bins and lower the effective \(m_i\) below \(K_{\rm seg}\) by the standard overlap factor. We therefore estimate \(m_i\) empirically from repeats via \(\widehat m_i=(\bar S_i/s_i)^2\) rather than equating it with the nominal segment count, and replace the diagonal channel by the covariance-weighted one of \eqref{eq:cov-kl} when overlap or short records make the bins dependent. The correlated-bin ablation of \cref{sec:validation} quantifies how much that dependence changes the audit.

\subsection{Averaged PSD likelihood}

With \(S_i(\theta)=S_\theta(k_i)\) and \(r_i=S_i(\theta')/S_i(\theta)\), the joint gamma law gives the closed-form divergences (derived in \cref{app:gamma})
\begin{equation}
\KL(P_\theta\Vert P_{\theta'})=\sum_i m_i\!\left[\log r_i+\tfrac{1}{r_i}-1\right],
\label{eq:gamma-kl}
\end{equation}
\begin{equation}
C(\theta,\theta')=\sup_{0\le s\le1}\sum_i m_i\!\left[(1-s)\log S_i(\theta)+s\log S_i(\theta')+\log\!\Big(\tfrac{1-s}{S_i(\theta)}+\tfrac{s}{S_i(\theta')}\Big)\right].
\label{eq:chernoff}
\end{equation}
For correlated bins, with \(Z_i=\log\widehat S_i\), \(\mu_i=\log S_\theta(k_i)\), and \(\widehat\Sigma\) estimated from repeats,
\begin{equation}
\KL(P_\theta^R\Vert P_{\theta'}^R)\approx\tfrac12(\mu(\theta)-\mu(\theta'))^\top\widehat\Sigma^{-1}(\mu(\theta)-\mu(\theta')),
\label{eq:cov-kl}
\end{equation}
with Chernoff equal to one quarter of this for equal covariance; all formulas hold with \(\operatorname{diag}(m_i)\) replaced by \(\widehat\Sigma^{-1}\). Estimating the full \(n\times n\) covariance carries its own sample cost: the sample covariance from \(r\) repeats is singular for \(r\le n\) and is reliably invertible only for \(r\gg n\), so in the common regime \(r\lesssim n\) we shrink \(\widehat\Sigma\) toward a diagonal target \cite{ledoit2004} and accept a controlled bias, whereas the diagonal-gamma channel needs only the \(n\) per-bin variances and stays usable at small \(r\).

\section{Recipe-release game and adversary classes}
\label{sec:game}

With the channel in hand, we now name the opponent (\cref{fig:threat-model}). We state the release game, fix the budgets that bound a profiled adversary, and prove the four guarantees the audit will quote: the optimal pairwise test, the finite-library threshold, the protected-bit advantage from released KL, and the finite-resource corrections. These classes form a statistical audit taxonomy graded by profiling effort.

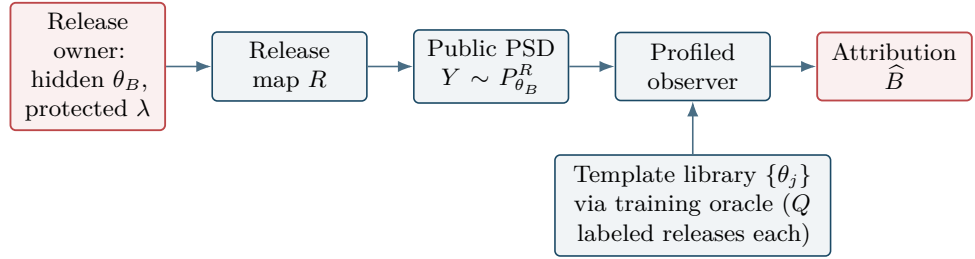
\begin{figure}[htbp]
\centering
\begin{tikzpicture}[
  >=Latex, font=\footnotesize,
  box/.style={draw=euvAcol, line width=.6pt, rounded corners=2pt, align=center,
              inner sep=3.5pt, fill=euvAcol!6, text width=18mm, minimum height=9mm},
  sec/.style={draw=euvBcol, line width=.7pt, rounded corners=2pt, align=center,
              inner sep=3.5pt, fill=euvBcol!8, text width=18mm, minimum height=9mm},
  flow/.style={-{Latex[length=2mm]}, line width=.6pt, draw=euvAcol!80},
  node distance=7mm and 6mm,
]
\node[sec] (owner) {Release owner: hidden $\theta_B$, protected $\lambda$};
\node[box, right=of owner] (map) {Release map $R$};
\node[box, right=of map] (pub) {Public PSD $Y\sim P_{\theta_B}^{R}$};
\node[box, right=of pub] (obs) {Profiled observer};
\node[sec, right=of obs] (bit) {Attribution $\widehat B$};
\node[box, below=of obs, text width=34mm] (lib)
  {Template library $\{\theta_j\}$ via training oracle ($Q$ labeled releases each)};
\draw[flow] (owner) -- (map);
\draw[flow] (map) -- (pub);
\draw[flow] (pub) -- (obs);
\draw[flow] (obs) -- (bit);
\draw[flow] (lib) -- (obs);
\end{tikzpicture}
\caption{\textbf{Threat model for a metrology-release audit.} A release owner holds a hidden recipe \(\theta_B\) whose protected coordinate is the effective acid-transport length \(\lambda\); the release map \(R\) exposes a finite-band PSD transcript \(Y\sim P_{\theta_B}^{R}\). A profiled observer, trained by a labeled-template oracle (\(Q\) releases per recipe), attributes the protected bit \(\widehat B\) by likelihood. The profiling budget \((N,Q,L,T)\) (library size, training depth, challenge length, and score evaluations) bounds the observer's effort, and the audit reports how distinguishable the protected pair is under that budget.}
\label{fig:threat-model}
\end{figure}

Concretely, the release owner is a fab or resist supplier that publishes roughness PSDs for benchmarking, metrology quality control, or a data-sharing consortium; the profiled observer is a competitor or analyst who can measure reference wafers of \emph{known} recipes (or holds previously disclosed recipe--PSD pairs) and uses them as the template library; and the protected bit is a recipe coordinate the owner would not disclose, such as quencher loading or acid-transport length. The audit asks whether the published statistic lets such an observer recover that bit, and at what released-band and sample-size cost.

\begin{experiment}[Statistical side-channel recipe-release distinguishability]
\label{exp:release}
Given \((\Theta,R,u,b)\): \textbf{(1)} the adversary receives budgets \((N,Q,L,T)\); \textbf{(2)} trains on \(Q_j\) releases per selected template; \textbf{(3)} a utility-equivalent protected pair \((\theta_0,\theta_1)\) is chosen, a uniform bit \(B\) sampled, and \(L\) releases from \(P_{\theta_B}^R\) given; \textbf{(4)} outputs \(\widehat B\) with advantage \(\Adv_R=|\Pr[\widehat B=B]-\tfrac12|\).
\end{experiment}

A protected pair is \emph{utility-equivalent} when both recipes meet the same public specification (identical target line-edge-roughness and critical-dimension budgets, so a downstream user cannot prefer one on utility grounds) yet differ in the protected coordinate \(\lambda\); the audit asks whether the release nonetheless separates them. The policy objects \(u\) and \(b\) encode this governance choice: \(u\) is the utility every release must preserve and \(b\) is the bit the owner intends to keep private.

\begin{definition}[Configuration and classes]
\label{def:classes}
The configuration is \(\kappa=(n,K,\{k_i,m_i\},L,N,Q,T,\delta_{\rm alg},R,u,b,\Theta_0)\). A score query returns \(\widetilde\ell_j=\sum_\ell\log q_{\widehat\theta_j}^R(Y_\ell)+e_j\), \(|e_j|\le\delta_{\rm alg}\). A finite-compute rule scores a data-dependent set \(\calM(Y)\subseteq\Theta_0\) with \(|\calM(Y)|\le M\le\lfloor T/(cnL)\rfloor\), since scoring one template over \(L\) releases and \(n\) bins costs \(cnL\) elementary operations; its coverage error is \(\pi_{\rm miss}=\sup_{\theta\in\Theta_0}\Pr_\theta[\theta\notin\calM(Y)]\), the worst-case probability that the true template is never scored.
\end{definition}

The analysis supports five concentric adversary classes, in increasing realism. The \emph{ideal-template} adversary knows every released law \(P_{\theta_j}^R\) exactly and runs full-library maximum likelihood and pairwise likelihood-ratio testing; it is the information-theoretic baseline. The \emph{finite-training} adversary instead estimates each template from \(Q_j\) labeled releases and scores with the plug-in likelihood, paying the score perturbation of \cref{prop:bounds}. The \emph{finite-library} adversary stores a codebook of \(N\) recipes and attributes against it under the union-bound threshold built from \(C_{\min}\). The \emph{finite-compute} adversary may evaluate at most \(T\) likelihood scores, so it searches an adaptive sublibrary of size \(M\le\lfloor T/(cnL)\rfloor\) and pays an additional coverage term. Finally the \emph{mismatched-model} adversary scores under a fitted law \(Q_\theta^R\) while data are drawn from \(P_\theta^R\), and is controlled by a positive-margin, sub-exponential concentration bound. Each class is a strict relaxation of the one before it, so an upper bound proved for the ideal-template adversary specializes to all of them and the audit reports, row by row, where a given release sits against each. Throughout, these budgets bound the adversary's statistical profiling effort: how many labeled spectra it trains on, how many templates it stores, and how many likelihood scores it evaluates.

\begin{theorem}[Optimal pairwise test and finite-library threshold]
\label{thm:binary}
For ideal templates the optimal advantage after \(L\) releases is \(\tfrac12\TV((P_{\theta_0}^R)^{\otimes L},(P_{\theta_1}^R)^{\otimes L})\), with Bayes error \(\exp\{-LC_R+o(L)\}\). For a library of \(N\), exhaustive ML attribution obeys \(\Pr[\widehat\theta\ne\theta]\le(N-1)\exp(-LC_{\min})\), \(C_{\min}=\min_{a\ne b}C_R(\theta_a,\theta_b)\); a finite-compute search scoring \(M\) templates obeys the bound with \(M-1\), plus its coverage error.
\end{theorem}

\begin{proof}
For two simple hypotheses with equal prior probabilities, the Neyman--Pearson likelihood-ratio rule is Bayes optimal. If \(P_L=(P_{\theta_0}^R)^{\otimes L}\) and \(Q_L=(P_{\theta_1}^R)^{\otimes L}\), its success probability is \(\tfrac12(1+\TV(P_L,Q_L))\), so its advantage is \(\tfrac12\TV(P_L,Q_L)\). The Chernoff theorem for independent repetitions gives
\[
-\frac1L\log P_{e,L}\to C_R(\theta_0,\theta_1),
\]
hence \(P_{e,L}=\exp\{-LC_R+o(L)\}\).

For the library bound, condition on the true template being \(a\). Exhaustive ML makes an error only if some \(b\ne a\) has likelihood at least that of \(a\). For any fixed competitor and any \(s\in[0,1]\),
\[
\Pr_a\{\ell_b(Y^L)\ge\ell_a(Y^L)\}
\le
\E_a\exp\{s(\ell_b-\ell_a)\}
=
\left(\int (p_a^R)^{1-s}(p_b^R)^s\right)^L .
\]
Optimizing over \(s\) gives \(\Pr_a\{\ell_b\ge\ell_a\}\le \exp\{-LC_R(\theta_a,\theta_b)\}\). A union bound over the \(N-1\) incorrect templates yields \((N-1)\exp(-LC_{\min})\). If a finite-compute procedure scores only a set \(\calM(Y)\) with \(|\calM|\le M\), the same argument applies on the event that the true template is covered, giving \(M-1\) competitors; adding the coverage failure probability gives the final claim.
\end{proof}

\begin{theorem}[Protected-bit advantage from released KL]
\label{thm:protected-bit-adv}
If \(\KL(P_{\theta_0}^R\Vert P_{\theta_1}^R)\le\eps\) and \(\KL(P_{\theta_1}^R\Vert P_{\theta_0}^R)\le\eps\), then every adversary in \(\calA_\kappa\) has \(\Adv_R\le\tfrac12\sqrt{L\eps/2}\).
\end{theorem}

\begin{proof}
Let \(P_i=P_{\theta_i}^R\). Independence gives
\[
\KL(P_0^{\otimes L}\Vert P_1^{\otimes L})=L\KL(P_0\Vert P_1)\le L\eps
\]
and the same bound in the reverse direction. Pinsker's inequality therefore gives
\(\TV(P_0^{\otimes L},P_1^{\otimes L})\le\sqrt{L\eps/2}\). Any adversary's decision rule is a measurable set (or a randomized mixture of such sets), so its distinguishing advantage is bounded by one half of the total variation distance. The optimal likelihood-ratio test attains the total-variation benchmark, hence every restricted adversary in \(\calA_\kappa\) obeys \(\Adv_R\le\tfrac12\sqrt{L\eps/2}\).
\end{proof}

\begin{remark}[The Chernoff exponent is the operative quantity]
The Pinsker bound of \cref{thm:protected-bit-adv} is the convenient closed form to quote when only the KL budget \(\eps\) is known, but it is loose. When the released laws are known the operative quantity is the Chernoff exponent: the optimal test of \cref{thm:binary} has error \(\exp\{-LC_R+o(L)\}\) and advantage \(\tfrac12(1-2P_{e,L})\), so the audit reports \(C_R\) (and the derived \(L_{0.9}\)) as the exact exponential rate and falls back on Pinsker only when a release is summarized by a divergence budget alone.
\end{remark}

\begin{proposition}[Finite-compute, finite-training, mismatch]
\label{prop:bounds}
(i) A depth-\(T\) tree scoring \(\calM(Y)\), \(|\calM|\le M\), coverage error \(\pi_{\rm miss}\), obeys \(\Pr[\widehat\theta\ne\theta]\le\pi_{\rm miss}+(M-1)\exp\{-LC_{\rm score}+2\delta_{\rm alg}\}\). (ii) Under \cref{ass:psd}.A2 and A4 (\(\sqrt{Q_j}\)-consistent, efficient template estimators with a per-release score and Hessian bounded in probability on a neighbourhood of \(\theta_j\), i.e.\ local asymptotic normality), the plug-in score perturbation is \(O_p(\sqrt{L/Q_j})+O_p(L/Q_j)\), uniformly \(O_p(\sqrt{L\log N}\max_jQ_j^{-1/2}+L\max_jQ_j^{-1})\). (iii) With positive mismatch margins \(\Gamma=\min\{\Gamma_0,\Gamma_1\}>0\) and sub-exponential ratios, \(P_e(L)\le\exp[-cL\min\{\Gamma^2/\nu^2,\Gamma/b\}]\).
\end{proposition}

\begin{proof}
(i) On the event that the true template belongs to the scored set, an approximate score error can reverse a pairwise comparison between the true template \(a\) and a competitor \(b\) only if the exact log-likelihood gap satisfies \(\ell_b-\ell_a\ge-2\delta_{\rm alg}\). The Chernoff calculation in the proof of \cref{thm:binary}, with this shifted threshold, gives
\[
\Pr_a\{\widetilde\ell_b\ge\widetilde\ell_a\}
\le \exp\{-LC_R(a,b)+2\delta_{\rm alg}\}
\le \exp\{-LC_{\rm score}+2\delta_{\rm alg}\}.
\]
Union bounding over at most \(M-1\) scored competitors and adding the probability \(\pi_{\rm miss}\) that the true template is not scored proves (i).

(ii) Write \(\Delta_j=\widehat\theta_j-\theta_j\). Efficiency gives \(\Delta_j=O_p(Q_j^{-1/2})\). Under the regularity conditions in \cref{ass:psd}, Taylor expansion of the \(L\)-sample score around \(\theta_j\) gives
\[
\sum_{\ell=1}^L\log q_{\theta_j+\Delta_j}(Y_\ell)
-\sum_{\ell=1}^L\log q_{\theta_j}(Y_\ell)
=
\Delta_j^\top\sum_{\ell=1}^Ls_j(Y_\ell)
+\frac12\Delta_j^\top\!\left(\sum_{\ell=1}^LH_j(\widetilde\theta_j;Y_\ell)\right)\!\Delta_j .
\]
The centered score sum is \(O_p(\sqrt L)\) and the Hessian sum is \(O_p(L)\) on the local likelihood neighborhoods used by the plug-in templates. Therefore the perturbation for template \(j\) is \(O_p(\sqrt{L/Q_j})+O_p(L/Q_j)\). Taking a maximum over \(N\) templates replaces the score-sum scale by \(O_p(\sqrt{L\log N})\), yielding the displayed uniform bound.

(iii) For a mismatched binary score, let \(X_\ell\) be the per-release log-score advantage of the correct protected bit over the incorrect one. The positive-margin assumption says \(\E X_\ell\ge\Gamma\), and the centered variables \(X_\ell-\E X_\ell\) are sub-exponential with parameters \((\nu^2,b)\). An error implies \(\sum_{\ell=1}^LX_\ell\le0\), hence
\[
\sum_{\ell=1}^L(X_\ell-\E X_\ell)\le-L\Gamma .
\]
Bernstein's inequality for sub-exponential variables gives
\[
\Pr\!\left\{\sum_{\ell=1}^L(X_\ell-\E X_\ell)\le-L\Gamma\right\}
\le \exp[-cL\min\{\Gamma^2/\nu^2,\Gamma/b\}],
\]
for a universal constant \(c>0\), which proves the mismatch claim.
\end{proof}

\begin{remark}[Verifying the mismatch margin]
The positive-margin hypothesis of (iii) is checkable rather than assumed: on held-out labeled releases, form the empirical per-release log-likelihood-ratio advantage of the correct protected bit over the fitted templates, average it to \(\widehat\Gamma\), and certify \(\Gamma>0\) when the one-sided confidence lower bound of \(\widehat\Gamma\) is positive. A nonpositive bound flags that the fitted model does not separate the protected pair, so the released exponent does not apply to that profiler. This is step P7 of \cref{sec:protocol}.
\end{remark}

\section{Composition and local geometry}

\begin{theorem}[Additive release accounting]
\label{thm:composition}
Conditionally independent releases with KL budgets \(\eps_t\) compose to budget \(\sum_t\eps_t\), and locally \(\calI_{\oplus_t R_t}=\sum_t\calI_{R_t}\).
\end{theorem}

\begin{proof}
Let \(P_{\theta,t}\) and \(P_{\theta',t}\) be the laws of release \(t\). Conditional independence makes the joint laws products, so
\[
\KL\!\left(\bigotimes_tP_{\theta,t}\,\middle\Vert\,\bigotimes_tP_{\theta',t}\right)
=\sum_t\KL(P_{\theta,t}\Vert P_{\theta',t})
\le\sum_t\eps_t .
\]
For the local statement, set \(\theta'=\theta+h\) and expand each term:
\[
\KL(P_{\theta,t}\Vert P_{\theta+h,t})
=\frac12 h^\top\calI_{R_t}(\theta)h+O(\|h\|^3).
\]
Summing over \(t\) and identifying the quadratic form of the product release gives \(\calI_{\oplus_tR_t}=\sum_t\calI_{R_t}\).
\end{proof}

For \(f_\theta=\log S_\theta\) and small \(h\), \(\KL=\tfrac12 h^\top\calI_R h+O(\norm h^3)\), with \(\calI_K(\theta)=\int_0^K w\,\nabla f\nabla f^\top\dd k\) and \(\inner{g}{h}_{w,K}=\int_0^K w\,gh\,\dd k\). A tangent has zero local exponent under repetition iff it lies in \(\Null(\calI_R)\).

\section{Finite-band EUV transport leakage}
\label{sec:transport}

We now extract the physics. The plan is direct: eliminate amplitude and blur as nuisance directions, expand what remains of the transport tangent at small band edge, and show that its first surviving piece is a fourth-order curvature whose finite-band norm grows as $K^9$, the leakage law that drives every threshold in the paper.

For negligible \(S_0\), \(\partial_\alpha f=1\), \(\partial_\beta f=-k^2\), \(\partial_\lambda f=g_\lambda=-3\lambda k^2/(1+\lambda^2k^2)\), and the protected information is the Schur complement
\begin{equation}
\calI_{\lambda\mid\alpha,\beta}(K)=\inner{g_\lambda}{(I-\Pi_{\alpha,\beta})g_\lambda}_{w,K},
\label{eq:conditional-transport}
\end{equation}
with \(\Pi_{\alpha,\beta}\) projecting onto \(\operatorname{span}\{1,k^2\}\).

\begin{theorem}[Transport-knee leakage exponent]
\label{thm:k9}
For flat \(w\) and \(K\lambda\ll1\),
\begin{equation}
\calI_{\lambda\mid\alpha,\beta}(K)=\frac{64}{1225}\,w\lambda^6K^9+O(w\lambda^8K^{11}),
\label{eq:k9}
\end{equation}
so for \(\lambda_\pm=\lambda\pm\Delta\lambda/2\), \(C_R=\tfrac18(\Delta\lambda)^2\calI_{\lambda\mid\alpha,\beta}(K)+O((\Delta\lambda)^3)\).
\end{theorem}

\begin{proof}
Let \(\calB=\operatorname{span}\{1,k^2\}\) and \(R_\calB=I-\Pi_{\alpha,\beta}\). For \(K\lambda\ll1\),
\[
g_\lambda(k)=-\frac{3\lambda k^2}{1+\lambda^2k^2}
=-3\lambda k^2+3\lambda^3k^4+O(\lambda^5k^6).
\]
Since \(k^2\in\calB\), the first term is annihilated by \(R_\calB\), and therefore
\[
R_\calB g_\lambda=3\lambda^3R_\calB k^4+O(\lambda^5R_\calB k^6).
\]
The cross term between these two displayed terms is \(O(w\lambda^8K^{11})\), because \(\|R_\calB k^4\|^2=O(K^9)\) and \(\|R_\calB k^6\|^2=O(K^{13})\); the square of the remainder is of higher order. It remains to compute \(\|R_\calB k^4\|^2\).

Write the projection of \(k^4\) onto \(\calB\) as \(c_0+c_2k^2\). With \(\mu_j=\int_0^Kk^j\,\dd k=K^{j+1}/(j+1)\), the normal equations are
\[
\begin{pmatrix}
\mu_0&\mu_2\\
\mu_2&\mu_4
\end{pmatrix}
\begin{pmatrix}c_0\\c_2\end{pmatrix}
=
\begin{pmatrix}\mu_4\\\mu_6\end{pmatrix}.
\]
The determinant is \(K^6(1/5-1/9)=4K^6/45\), so Cramer's rule gives
\[
c_0=-\frac{3}{35}K^4,\qquad c_2=\frac67K^2 .
\]
By orthogonality,
\[
\|R_\calB k^4\|^2
=\int_0^Kk^8\,\dd k-c_0\int_0^Kk^4\,\dd k-c_2\int_0^Kk^6\,\dd k
=\left(\frac19+\frac{3}{175}-\frac{6}{49}\right)K^9
=\frac{64}{11025}K^9 .
\]
Consequently
\[
\calI_{\lambda\mid\alpha,\beta}(K)
=w\|R_\calB g_\lambda\|^2
=9\lambda^6w\,\frac{64}{11025}K^9+O(w\lambda^8K^{11})
=\frac{64}{1225}w\lambda^6K^9+O(w\lambda^8K^{11}).
\]
For the pair \(\lambda_\pm=\lambda\pm\Delta\lambda/2\), the local Chernoff expansion of a regular one-parameter likelihood after Schur elimination is \(C_R=\tfrac18(\Delta\lambda)^2\calI_{\lambda\mid\alpha,\beta}(K)+O((\Delta\lambda)^3)\), which gives the final statement.
\end{proof}

\begin{corollary}[Quencher exponent and safe band]
\label{cor:qbit}
\(\delta\lambda=-\tfrac12\lambda^3(k_q/D_H)\delta q_0\) gives \(\calI_{q_0\mid\alpha,\beta}(K)=\tfrac{16}{1225}w(k_q^2/D_H^2)\lambda^{12}K^9+O(\cdots)\), and the band edge \(K\le\big(\tfrac{1225\cdot8\eps}{64\,Lw\lambda^6(\Delta\lambda)^2}\big)^{1/9}\) keeps the local pairwise exponent below \(\eps\).
\end{corollary}

\begin{proof}
From \(\lambda^2=D_H/(k_qq_0+k_{\rm loss})\),
\[
\frac{\partial\lambda}{\partial q_0}
=-\frac12\lambda^3\frac{k_q}{D_H}.
\]
The Fisher information transforms by the chain rule, so
\[
\calI_{q_0\mid\alpha,\beta}
=\left(\frac{\partial\lambda}{\partial q_0}\right)^2
\calI_{\lambda\mid\alpha,\beta}
=\frac14\lambda^6\frac{k_q^2}{D_H^2}\cdot
\frac{64}{1225}w\lambda^6K^9+O(\cdots)
=\frac{16}{1225}w\frac{k_q^2}{D_H^2}\lambda^{12}K^9+O(\cdots).
\]
For a protected pair separated by \(\Delta\lambda\), the \(L\)-release local exponent is
\[
LC_R=\frac{L}{8}(\Delta\lambda)^2\frac{64}{1225}w\lambda^6K^9+O(LK^{11}).
\]
Ignoring the higher-order term in the low-band regime and requiring \(LC_R\le\eps\) gives the displayed ninth-root bound on \(K\).
\end{proof}

\subsection{Exact finite-band matrix}

With \(X=K\lambda\) and flat \(w\), the exact entries (derived term by term in \cref{app:fisher}) are
\begin{align}
I_{\alpha\alpha}&=wK,\quad I_{\alpha\beta}=-w\tfrac{K^3}{3},\quad I_{\beta\beta}=w\tfrac{K^5}{5},\label{eq:fisher-basic}\\
I_{\alpha\lambda}&=-\tfrac{3w}{\lambda^2}(X-\arctan X),\quad
I_{\beta\lambda}=\tfrac{3w}{\lambda^4}\big(\tfrac{X^3}{3}-X+\arctan X\big),\label{eq:fisher-cross}\\
I_{\lambda\lambda}&=\frac{9w}{\lambda^3}\Big(X-\tfrac32\arctan X+\tfrac{X}{2(1+X^2)}\Big),
\label{eq:fisher-transport}
\end{align}
and \(\calI_{\lambda\mid\alpha,\beta}\) is the Schur complement of the \((\alpha,\beta)\) block.

\begin{remark}[Corrected prefactor]
\label{rem:correction}
Equation \eqref{eq:fisher-transport} carries \(9w/\lambda^3\): substituting \(u=\lambda k\) in \(9\lambda^2w\int_0^K k^4(1+\lambda^2k^2)^{-2}\dd k\) gives \((9w/\lambda^3)\int_0^X u^4(1+u^2)^{-2}\dd u\), with small-\(X\) limit \((9/5)w\lambda^2K^5\) (see \cref{app:fisher}). An earlier draft printed \(9w/\lambda\); the package pins both the entry and its limit in \texttt{test\_fisher\_transport\_knee}. The headline law \eqref{eq:k9} is unaffected, being the independently derived residual norm.
\end{remark}

\subsection{Chemical identifiability}

A single spectrum identifies only \(\mu=\lambda^{-2}=(k_qq_0+k_{\rm loss})/D_H\); a quencher sweep identifies the affine line \(\lambda^{-2}=aq_0+b\), \(a=k_q/D_H\), \(b=k_{\rm loss}/D_H\). Two chemistries coincide iff \((k_qq_0+k_{\rm loss})/D_H=(k_q'q_0'+k_{\rm loss}')/D_H'\).

\section{Release maps and controlled disclosure}
\label{sec:release}

Having quantified the leak, we now design against it. We ask which public statistic a release owner should publish: we require the protected tangent to lie in the nullspace of the release, then maximize the utility retained inside that nullspace, ending with the explicit minimum-leakage statistic under a utility constraint.

For \(Z=Rf_\theta+\xi\), \(\xi\sim N(0,\Sigma_R)\), the released Fisher is \(\calI_R=G_\theta^\top R^\top\Sigma_R^{-1}RG_\theta\).

\begin{theorem}[Nullspace and zero-leakage utility projection]
\label{thm:projection-synthesis}
A protected tangent \(p\) has zero local exponent iff \(RG_\theta p=0\). Among \(d\)-dimensional projection releases, zero leakage holds iff \(\calZ\subseteq\calP^\perp\), and the maximizer of retained utility energy is the span of the top \(d\) eigenfunctions of \(\Pi_{\calP^\perp}\Pi_\calU\Pi_{\calP^\perp}\).
\end{theorem}

\begin{proof}
For the Gaussian release \(Z=Rf_\theta+\xi\), \(\xi\sim N(0,\Sigma_R)\), a local displacement \(h\) changes the mean by \(RG_\theta h\). The local KL quadratic form is
\[
\frac12 h^\top G_\theta^\top R^\top\Sigma_R^{-1}RG_\theta h.
\]
Because \(\Sigma_R^{-1}\) is positive definite on the released coordinates, this quadratic form vanishes for \(h=p\) iff \(RG_\theta p=0\).

For projection releases, write the released subspace as \(\calZ\). Zero leakage for every protected tangent in \(\calP\) means every released direction is orthogonal to \(\calP\), equivalently \(\calZ\subseteq\calP^\perp\). If the utility subspace is \(\calU\), the retained utility energy of a rank-\(d\) projection \(\Pi_\calZ\) is \(\Tr(\Pi_\calZ\Pi_\calU)\). Under the zero-leakage constraint this equals
\[
\Tr\!\left(\Pi_\calZ\,\Pi_{\calP^\perp}\Pi_\calU\Pi_{\calP^\perp}\right).
\]
Ky Fan's maximum principle says that this trace over rank-\(d\) projections is maximized by taking \(\calZ\) to be the span of the top \(d\) eigenfunctions of the positive semidefinite operator \(\Pi_{\calP^\perp}\Pi_\calU\Pi_{\calP^\perp}\).
\end{proof}

\begin{theorem}[Rank-one leakage-utility optimizer]
\label{thm:minimax}
For unit \(p,u\) with \(\rho=\inner{p}{u}\), isotropic noise \(\tau^2\), and budget \(\inner{z}{p}^2/(2\tau^2)\le\eps\), the maximum utility is \(\sqrt{2\tau^2\eps}\,|\rho|+\sqrt{1-2\tau^2\eps}\sqrt{1-\rho^2}\) for \(2\tau^2\eps\le\rho^2\), and \(1\) (at \(z=u\)) once \(2\tau^2\eps\ge\rho^2\).
\end{theorem}

\begin{proof}
Only the plane spanned by \(p\) and \(u\) can matter: any component of \(z\) orthogonal to both consumes norm without increasing utility or leakage. Choose \(q\perp p\), \(\|q\|=1\), such that \(u=\rho p+\sqrt{1-\rho^2}\,q\), and write \(z=ap+bq\) with \(a^2+b^2\le1\). The leakage constraint is \(|a|\le A:=\sqrt{2\tau^2\eps}\). Aligning signs with \(u\), the utility for fixed \(a\) is maximized by \(b=\sqrt{1-a^2}\), giving
\[
\phi(a)=a|\rho|+\sqrt{1-a^2}\sqrt{1-\rho^2},\qquad 0\le a\le \min\{A,1\}.
\]
The derivative is
\[
\phi'(a)=|\rho|-\frac{a}{\sqrt{1-a^2}}\sqrt{1-\rho^2},
\]
so the unconstrained maximizer is \(a=|\rho|\), where \(z=u\) and \(\phi=1\). If \(A\ge|\rho|\), this point is feasible and the optimum is \(1\). If \(A<|\rho|\), \(\phi\) is increasing on \([0,A]\), so the optimum saturates the leakage constraint:
\[
\phi(A)=A|\rho|+\sqrt{1-A^2}\sqrt{1-\rho^2}
=\sqrt{2\tau^2\eps}\,|\rho|+\sqrt{1-2\tau^2\eps}\sqrt{1-\rho^2}.
\]
\end{proof}

\section{Dose-to-event normalization and calibrated scale}

A \SI{13.5}{nm} photon has energy \(\SI{91.8}{eV}\); a dose of \(E\) (\si{mJ/cm^2}) deposits \(N_\gamma(E)=0.6797E\) photons per \si{nm^2}, an absorbed fraction \(A_T\) and acid yield \(\eta\) give \(n_a(E)=0.6797E\,A_T\eta\) acid events per \si{nm^2}, and an edge cell of length \(\ell\) and width \(\lambda\) holds \(N_{\rm cell}\approx n_a\ell\lambda\) events. The synthetic audits are anchored to standard EUV magnitudes by these relations: dose \(E\in[20,60]\,\si{mJ/cm^2}\) and absorbed fraction \(A_T\in[0.1,0.6]\) fix the photon and acid-event counts; acid yield \(\eta\in[0.5,5]\) and an LER \(3\sigma\in[2,5]\,\mathrm{nm}\) fix the stochastic amplitude; the protected transport length \(\lambda\in[2,10]\,\mathrm{nm}\) and Gaussian blur \(\sqrt{\beta}\in[1,5]\,\mathrm{nm}\) fix the spectral shape and the knee position \(K\lambda\); and the Welch shape \(m_i\in[8,32]\) with band edge \(K\in[0.05,1.5]\,\mathrm{nm}^{-1}\) fix the per-bin variance \(S_i^2/m_i\) and whether the knee is inside the released band. As a worked point, \(E=30\), \(A_T=0.3\), \(\eta=2\) give \(\approx12.2\) acid events per \si{nm^2}; with \(\lambda=\ell=5\,\mathrm{nm}\) a transport cell then carries \(\approx305\) events: large enough for a meaningful PSD estimate, small enough for stochastic release exponents to matter.

\section{A practical audit protocol for a measured PSD release}
\label{sec:protocol}

The theory above is operational. Given a measured PSD release with repeats, the following protocol returns the numbers the audit quotes (the released exponent, the safe band, and the diagnostics that decide whether either can be trusted). Each step instantiates a result proved earlier, and the synthetic and calibrated audits of \crefrange{sec:audits}{sec:validation} are worked instances of it.

\begin{enumerate}[label=\textbf{P\arabic*.}]
\item \textbf{Effective degrees of freedom.} From \(r\) repeated spectra estimate the per-bin mean \(\bar S_i\) and standard deviation \(s_i\), and set the effective Welch shape \(\widehat m_i=(\bar S_i/s_i)^2\) (equivalently \(2\widehat m_i\) effective degrees of freedom per bin). Test \cref{ass:psd}.A1 by a per-bin goodness-of-fit of the repeats against \(\mathrm{Gamma}(\widehat m_i,\bar S_i/\widehat m_i)\); a systematic shape departure means the gamma channel is the wrong likelihood for that bin.
\item \textbf{Bin covariance.} Form the log-residuals \(Z_i=\log\widehat S_i-\log\bar S_i\) and estimate their covariance \(\widehat\Sigma\) from the repeats. With few repeats relative to the number of bins, regularize by shrinkage toward a diagonal target \cite{ledoit2004} before inverting. If the off-diagonal mass is negligible the diagonal-gamma channel \eqref{eq:gamma-kl}--\eqref{eq:chernoff} applies; otherwise use the covariance-weighted channel \eqref{eq:cov-kl} with \(\operatorname{diag}(m_i)\) replaced by \(\widehat\Sigma^{-1}\).
\item \textbf{Screened-model fit and residuals.} Fit \eqref{eq:screened-psd} to \(\log\bar S_i\) by weighted least squares, obtaining \((\widehat A,\widehat\beta,\widehat\lambda,\widehat S_0)\). Report the log-PSD goodness-of-fit and inspect the residuals \(\log\bar S_i-\log S_{\widehat\theta}(k_i)\) for unmodeled structure such as a failed plateau, a misplaced knee, or an unaccounted floor; structure here invalidates \cref{ass:psd}.A2 and the tangent algebra that follows.
\item \textbf{Nuisance profiling and protected tangent.} Declare the nuisance coordinates (amplitude \(A\) and blur \(\beta\), and the floor \(S_0\) if it is not negligible) and form the nuisance span from the tangents \eqref{eq:floor-tangents}. Project the protected transport tangent \(g_\lambda\) off that span; the residual is the protected direction and its squared \(L^2_w\)-norm is the conditional information \eqref{eq:conditional-transport}, computed by the backward-stable QR projection used in the package rather than by inverting the Fisher matrix directly.
\item \textbf{Released exponent.} Evaluate the protected leakage exponent for the audited recipe pair: the exact gamma Chernoff \eqref{eq:chernoff} (or its covariance-weighted form), with the local check \(C_R\approx\tfrac18(\Delta\lambda)^2\calI_{\lambda\mid\alpha,\beta}(K)\) of \cref{thm:k9}. In the low band this is the ninth-order law \(\calI_{\lambda\mid\alpha,\beta}(K)=\tfrac{64}{1225}w\lambda^6K^9\).
\item \textbf{Sample size and safe band.} Report \(L_{0.9}\), the number of releases at which optimal attribution reaches \(90\%\) (from the Bayes-error rate \(\exp\{-LC_R\}\) of \cref{thm:binary}), and the safe-band edge \(K\le\big(\tfrac{1225\cdot 8\eps}{64\,Lw\lambda^6(\Delta\lambda)^2}\big)^{1/9}\) of \cref{cor:qbit} that keeps the local exponent below a chosen budget \(\eps\).
\item \textbf{Mismatch margin.} If templates are fitted rather than known exactly, certify a positive mismatch margin \(\Gamma>0\) and quote the finite-training and mismatch corrections of \cref{prop:bounds}(ii)--(iii); a nonpositive margin means the released bound does not hold for that profiler.
\item \textbf{Floor conditioning.} Tabulate the \(3\times3\) nuisance-block condition number over the released \((K\lambda,S_0/A)\) as in \cref{fig:coupling}. A floor-blind audit is trustworthy only in the well-conditioned corner of \cref{prop:floor-coupling}; where the conditioning is large the floor tangent \(\partial_{S_0}f\) must be profiled explicitly and the floor-aware exponent reported instead.
\end{enumerate}

The protocol fails safe. Steps P1--P3 are gates: a failed gamma fit, an ill-estimated covariance, or a structured residual stops the audit before any exponent is quoted, because each invalidates a stated channel condition. Steps P7--P8 bound the two ways a released exponent can overstate protection (an underprofiled or mismatched adversary, and an ill-conditioned floor subtraction), so the reported safe band is the one a release owner can defend.

Read from the release owner's side, the same steps are a disclosure checklist. Choose the released band so its edge \(K\) stays below the safe-band edge of \cref{cor:qbit} for the protection budget \(\eps\) and the worst-case challenge length \(L\) one is willing to grant (P5--P6); estimate the bin covariance from repeats so the audit does not \emph{under}-state leakage by wrongly assuming independence (P2); state the template threat model explicitly as the profiling budget \((N,Q,L,T)\) granted to the observer (\cref{def:classes}); profile the floor wherever the nuisance block is ill-conditioned (P8); and publish the resulting safe-band edge alongside the released statistic so a downstream auditor can reproduce the same exponent.

These steps are exactly what a deployment runs on repeated measured releases: a handful of repeats (\(r\gtrsim10\)) stabilizes the effective degrees of freedom, a full covariance needs \(r\ge n\), and the first failed gate (a rejected gamma fit, an ill-estimated covariance, or a structured screened-fit residual) halts the audit before any exponent is quoted, so the same procedure runs unchanged once measured EUV releases become available.

\section{Reconstructed and published-scale EUV PSD audits}
\label{sec:audits}

We now confront the theory with numbers. We first fit the screened model to a public-scale EUV PSD, then calibrate the transport audit to a reported 18-nm half-pitch measurement, establishing that the exponents we computed correspond to realistic released bands before we turn to controlled synthetic experiments.

\subsection{Reconstructed-PSD fit (pipeline smoke test)}

The original digitized points from the Naulleau--McClinton public EUV PSD (reported correlation length \(17\,\mathrm{nm}\)) \cite{naulleau2011metrology} were not available to us, so this subsection uses \emph{reconstructed} points only, regenerated from the manuscript's reported screened-model fit (see the provenance remark below); they exercise the fit pipeline and are not measured data. On these reconstructed points, \cref{fig:measured-psd-fit} fits \eqref{eq:screened-psd} (\(k=2\pi f/1000\,\mathrm{nm}^{-1}\)): the free-\(\lambda\) fit attains a small log-RMSE while forcing \(\lambda=17\,\mathrm{nm}\) raises it substantially, illustrating that \emph{within this model} a correlation-length descriptor and the screened transport coordinate are related but not interchangeable without fixing PSD convention, roughness exponent, blur, and floor.

\begin{remark}[Provenance]
The original slide and its digitization were unavailable; the points are \emph{reconstructed} from the reported fit with log-normal scatter and a fixed seed (\texttt{scripts/make\_reconstructed\_psd.py}). They exercise the pipeline and are not primary measured data.
\end{remark}

\begin{figure}[htbp]
\centering
\input{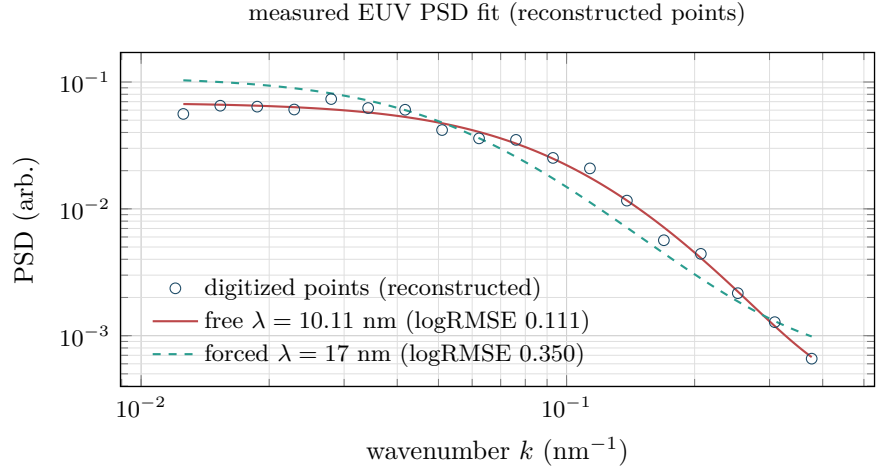}
\centering
\begin{tikzpicture}
\begin{axis}[
    euvaxis, width=0.7\linewidth, height=6.0cm,
    xmode=log, ymode=log,
    xlabel={wavenumber $k$ (nm$^{-1}$)}, ylabel={PSD (arb.)},
    legend pos=south west, legend cell align=left,
    title={\footnotesize measured EUV PSD fit (reconstructed points)},
]
\addplot[euvA, only marks, mark=o, mark size=2pt]
    table[x=k, y=psd] {figures/data/measured_points.dat};
\addlegendentry{digitized points (reconstructed)}
\addplot[euvB, thick]
    table[x=k, y=psd] {figures/data/measured_free.dat};
\addlegendentry{free $\lambda=\lamfree$ nm (logRMSE \rmsefree)}
\addplot[euvC, dashed, thick]
    table[x=k, y=psd] {figures/data/measured_forced.dat};
\addlegendentry{forced $\lambda=17$ nm (logRMSE \rmseforced)}
\end{axis}
\end{tikzpicture}
\caption{\textbf{Screened-model fit to a reconstructed EUV PSD.} Circles are the reconstructed digitized points; the solid curve is the unconstrained least-squares fit of \eqref{eq:screened-psd} in log-PSD (free \(\lambda\)), and the dashed curve forces \(\lambda=17\,\mathrm{nm}\), the reported correlation-length scale. The free fit follows the plateau, the transport roll-off, and the high-frequency floor with a small log-RMSE, whereas the forced curve is visibly too stiff through the knee. The gap between the two curves is the quantitative statement that a metrology correlation length and the screened transport coordinate \(\lambda\) are distinct descriptors: they agree only once PSD convention, roughness exponent, blur, and floor are pinned. A residual bootstrap \emph{on these reconstructed points} places the free \(\lambda\) near \(10\,\mathrm{nm}\) with a tight interval; this is an internal consistency statement about the reconstructed pipeline, not independent evidence for a physical transport length.}
\label{fig:measured-psd-fit}
\end{figure}

\subsection{Published 18-nm half-pitch scale (scale anchoring)}

This subsection sets the audit's scale: it ties the magnitudes (band edge, correlation length, averaging) to a reported 18-nm half-pitch measurement so the computed exponents refer to realistic numbers. Calibrated from \cite{mack2017level} (edge \(2400\,\mathrm{nm}\), \(\Delta y=5\,\mathrm{nm}\), correlation length \(7.2\,\mathrm{nm}\)): \cref{tab:published-scale} uses \(\lambda=7.2\,\mathrm{nm}\), \(K_{\rm Nyq}=\pi/\Delta y=0.628\,\mathrm{nm}^{-1}\), \(m=16\), \(\Delta\lambda/\lambda=5\%\), and the local Chernoff after Schur elimination of amplitude and blur.

\begin{table}[htbp]
\centering
\small
\renewcommand{\arraystretch}{1.15}
\caption{\textbf{Calibrated 18-nm half-pitch release audit.} Anchored to the reported public measurement (correlation length $7.2\,\mathrm{nm}$ used as a transport-scale proxy, edge length $2400\,\mathrm{nm}$ giving bin spacing $\Delta k=2\pi/2400\,\mathrm{nm}^{-1}$, sampling $\Delta y=5\,\mathrm{nm}$ giving Nyquist edge $\pi/\Delta y=0.628\,\mathrm{nm}^{-1}$), with $m=16$ effective averages and a protected separation $\Delta\lambda/\lambda=5\%$. The per-spectrum exponent is the local quadratic Chernoff $C=\tfrac18(\Delta\lambda)^2\,\mathcal{I}_{\lambda\mid\alpha,\beta}$ from the discrete Schur complement after amplitude and blur elimination; $L_{0.9}=\lceil\log(10)/C\rceil$.}
\label{tab:published-scale}
\setlength{\tabcolsep}{6pt}
\begin{tabular}{c c r c r l}
\toprule
{$K$ (nm$^{-1}$)} & {$K\lambda$} & {bins} & {$C$ per spectrum} & {$L_{0.9}$} & {severity} \\
\midrule
0.050 & 0.36 & 19 & $1.13\times10^{-6}$ & 2\,032\,914 & low \\
0.100 & 0.72 & 38 & $2.29\times10^{-4}$ & 10\,063 & low \\
0.200 & 1.44 & 76 & $1.30\times10^{-2}$ & 178 & moderate \\
0.400 & 2.88 & 152 & $1.48\times10^{-1}$ & 16 & high \\
0.628 & 4.52 & 240 & $3.76\times10^{-1}$ & 7 & high \\
\bottomrule
\end{tabular}
\par\vspace{3pt}
\begin{minipage}{0.92\linewidth}\footnotesize The same physical sample moves from \emph{low} to \emph{high} release severity as the band edge crosses the transport knee near $K\lambda\sim1$: a low-pass release below $K=0.1\,\mathrm{nm}^{-1}$ hides the protected bit behind $\sim\!10^{4}$ spectra, whereas the full Nyquist band exposes it in $\sim\!7$.\end{minipage}
\end{table}

\section{Synthetic gamma-PSD audit (model-internal evidence)}

Seed \(20260531\), 10{,}000 trials per cell with \(L_{0.9}\le5000\), Wilson 95\% intervals; \(n=128\) midpoint bins, \(m=16\), \(A=1\), \(S_0=0\), \(\beta=0.02\), \(\lambda=1\); \(\lambda_1=\lambda(1+\Delta)\), nuisance-matched by projecting the log \(\lambda\)-shift onto \(\operatorname{span}\{1,k^2\}\); exponent is the exact gamma Chernoff \eqref{eq:chernoff}.

\begin{table}[htbp]
\centering
\small
\renewcommand{\arraystretch}{1.15}
\caption{\textbf{Binary protected-bit challenge under the diagonal gamma channel.} For each released band $K\lambda$ and protected transport separation $\Delta\lambda/\lambda$, the exact Chernoff information $C$ of \eqref{eq:chernoff} is evaluated on the nuisance-matched pair (the log $\lambda$-shift projected off $\operatorname{span}\{1,k^2\}$, so amplitude and Gaussian blur give the adversary no help), and $L_{0.9}=\lceil\log(10)/C\rceil$ is the number of independent released spectra at which the optimal likelihood-ratio test first reaches $90\%$ success. The final column is the measured maximum-likelihood success over 1,000 seeded Monte-Carlo challenges with its Wilson $95\%$ interval, reported where $L_{0.9}\le5000$.}
\label{tab:mc-binary}
\setlength{\tabcolsep}{6pt}
\begin{tabular}{c c c r l}
\toprule
{$K\lambda$} & {$\Delta\lambda/\lambda$ (\%)} & {Chernoff $C$} & {$L_{0.9}$} & {ML success, 95\% CI} \\
\midrule
0.30 & 10 & $9.09\times10^{-6}$ & 253\,222 & -- \\
1.00 & 5 & $4.88\times10^{-3}$ & 473 & 0.993~{\footnotesize [0.986, 0.997]} \\
1.00 & 10 & $2.09\times10^{-2}$ & 111 & 0.984~{\footnotesize [0.974, 0.990]} \\
3.00 & 1 & $5.36\times10^{-3}$ & 430 & 0.986~{\footnotesize [0.977, 0.992]} \\
3.00 & 5 & $1.32\times10^{-1}$ & 18 & 0.984~{\footnotesize [0.974, 0.990]} \\
\bottomrule
\end{tabular}
\par\vspace{3pt}
\begin{minipage}{0.92\linewidth}\footnotesize Reading across a row, a wider band or a larger separation raises $C$ and collapses $L_{0.9}$; the Monte-Carlo column confirms the predicted operating point to within sampling error. The low-band cell ($K\lambda=0.30$) requires more than $2\times10^{5}$ spectra and is therefore quoted by its exponent alone.\end{minipage}
\end{table}

The RMS-only ablation releases the scalar \(R=\sum_i a_i\widehat S_i\), scored with the optimal Gaussian likelihood using \(\mu_\theta=\sum_ia_iS_i\), \(\sigma_\theta^2=\sum_ia_i^2S_i^2/m_i\). Intuitively this collapses the whole band to a single variance functional and discards the spectral \emph{shape} (plateau height, knee position, roll-off) where the transport coordinate lives, so two recipes matched in total released variance but differing in \(\lambda\) are nearly invisible to an RMS release; this is why the PSD-vs-RMS gap below is large.
\begin{equation}
\mu_\theta=\sum_i a_iS_i(\theta),\qquad \sigma_\theta^2=\sum_i a_i^2\frac{S_i(\theta)^2}{m_i}.
\label{eq:rms-gaussian}
\end{equation}

\begin{table}[htbp]
\centering
\small
\renewcommand{\arraystretch}{1.15}
\caption{\textbf{Full-band PSD versus RMS-only release on identical challenge samples.} Both columns score the same seeded gamma draws: the full-PSD column applies the optimal bin-wise likelihood-ratio test, while the RMS-only column releases the scalar $R=\sum_i\widehat S_i$ and applies the optimal Gaussian test using its exact first two moments \eqref{eq:rms-gaussian} (1,000 trials, Wilson $95\%$ CI). Because the protected transport coordinate lives in spectral \emph{shape} rather than total power, collapsing the band to one scalar destroys most of the attribution signal.}
\label{tab:rms-ablation}
\setlength{\tabcolsep}{6pt}
\begin{tabular}{c c r c c}
\toprule
{$K\lambda$} & {$\Delta\lambda/\lambda$ (\%)} & {$L$} & {full PSD ML} & {RMS-only ML} \\
\midrule
1.00 & 5 & 465 & 0.986~{\footnotesize [0.977, 0.992]} & 0.554~{\footnotesize [0.523, 0.585]} \\
1.00 & 10 & 109 & 0.986~{\footnotesize [0.977, 0.992]} & 0.550~{\footnotesize [0.519, 0.581]} \\
3.00 & 1 & 433 & 0.989~{\footnotesize [0.980, 0.994]} & 0.842~{\footnotesize [0.818, 0.863]} \\
3.00 & 5 & 18 & 0.984~{\footnotesize [0.974, 0.990]} & 0.814~{\footnotesize [0.789, 0.837]} \\
\bottomrule
\end{tabular}
\par\vspace{3pt}
\begin{minipage}{0.92\linewidth}\footnotesize At $K\lambda=1$ the RMS test sits near chance ($\approx0.55$) while the PSD test is essentially certain; only deep past the knee ($K\lambda=3$) does the scalar recover usable power, and even there the full PSD stays strictly stronger.\end{minipage}
\end{table}

Library thresholds use \(L\ge\log(N/\eta)/C\) at \(\eta=0.05\).

\begin{table}[htbp]
\centering
\small
\renewcommand{\arraystretch}{1.15}
\caption{\textbf{Finite-library attribution thresholds.} Number of independent released spectra $L$ at which exhaustive maximum-likelihood attribution against a stored codebook of $N$ recipes keeps the union-bound failure probability below $\eta=0.05$, i.e.\ $L\ge\log(N/\eta)/C_{\min}$ (\cref{thm:binary}). The same exact gamma Chernoff exponent $C$ that governs the binary challenge sets the scale, so the cost of enlarging the library grows only logarithmically while the released bandwidth enters through $C$.}
\label{tab:library}
\setlength{\tabcolsep}{8pt}
\begin{tabular}{c c r r r}
\toprule
{$K\lambda$} & {$\Delta\lambda/\lambda$ (\%)} & {$N=2^{10}$} & {$N=2^{15}$} & {$N=2^{20}$} \\
\midrule
0.30 & 5 & 4\,999\,126 & 6\,744\,395 & 8\,489\,665 \\
1.00 & 5 & 2\,036 & 2\,746 & 3\,457 \\
3.00 & 5 & 75 & 102 & 128 \\
3.00 & 10 & 20 & 26 & 33 \\
\bottomrule
\end{tabular}
\par\vspace{3pt}
\begin{minipage}{0.92\linewidth}\footnotesize Multiplying the library by $32$ (from $2^{10}$ to $2^{15}$) adds only $\log 32/C$ spectra: a constant offset, not a multiplicative one. Bandwidth, by contrast, changes $C$ by orders of magnitude and dominates the threshold.\end{minipage}
\end{table}

The floor-mismatch audit generates with an unmodeled floor \(S_0=\rho A\) while scoring no-floor templates.

\begin{table}[htbp]
\centering
\small
\renewcommand{\arraystretch}{1.15}
\caption{\textbf{Model-mismatch sensitivity of the binary test.} Challenge spectra are \emph{generated} with an unmodeled additive metrology floor $S_0=\rho A$, but \emph{scored} against the floor-free templates, so the test operates under a controlled misspecification at fixed $K\lambda=1$, $\Delta\lambda/\lambda=10\%$, and $L=109$ (1,000 trials, Wilson $95\%$ CI). This is the empirical face of the positive-margin guarantee (\cref{prop:bounds}): a bounded mismatch only erodes, but does not reverse, the likelihood margin.}
\label{tab:mismatch}
\setlength{\tabcolsep}{6pt}
\begin{tabular}{c c r c}
\toprule
{floor $\rho=S_0/A$} & {$K\lambda$} & {$L$} & {ML success, 95\% CI} \\
\midrule
0.00 & 1.00 & 109 & 0.986~{\footnotesize [0.977, 0.992]} \\
0.01 & 1.00 & 109 & 0.979~{\footnotesize [0.968, 0.986]} \\
0.03 & 1.00 & 109 & 0.971~{\footnotesize [0.959, 0.980]} \\
0.05 & 1.00 & 109 & 0.960~{\footnotesize [0.946, 0.970]} \\
\bottomrule
\end{tabular}
\par\vspace{3pt}
\begin{minipage}{0.92\linewidth}\footnotesize Success degrades smoothly and gently across the tested range: a few-percent unmodeled floor costs only about one to two points of success rather than collapsing the test, exactly the bounded-margin behaviour \cref{prop:bounds} predicts.\end{minipage}
\end{table}

\section{Model-internal numerical checks and additional computational evidence}
\label{sec:validation}

\subsection{Transport-knee numerical check}

\cref{fig:transport-knee} compares the exact flat-weight Schur complement (backward-stable QR projection) with \eqref{eq:k9}: the ratio \(\to1\) and log-log slope \(\to9\) as \(K\lambda\to0\), numerically checking consistency with the asymptotic \cref{thm:k9} and the corrected prefactor of \cref{rem:correction}.

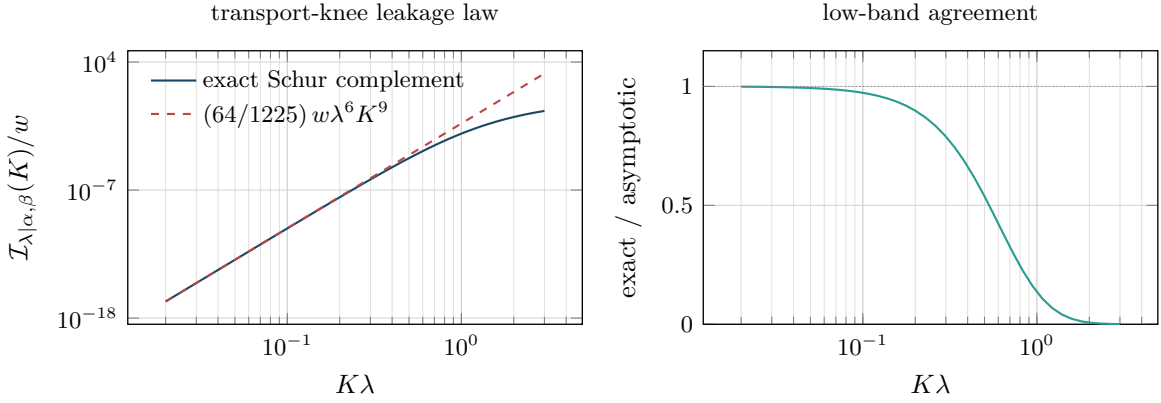
\begin{figure}[htbp]
\centering
\centering
\begin{tikzpicture}
\begin{groupplot}[
    group style={group size=2 by 1, horizontal sep=1.6cm},
    width=0.46\linewidth, height=5.2cm, euvaxis,
]
\nextgroupplot[
    xmode=log, ymode=log,
    xlabel={$K\lambda$}, ylabel={$\mathcal{I}_{\lambda\mid\alpha,\beta}(K)/w$},
    title={\footnotesize transport-knee leakage law},
    legend pos=north west, legend cell align=left,
]
\addplot[euvA, thick] table[x=Klam, y=exact] {figures/data/transport_knee.dat};
\addlegendentry{exact Schur complement}
\addplot[euvB, dashed, thick] table[x=Klam, y=asym] {figures/data/transport_knee.dat};
\addlegendentry{$(64/1225)\,w\lambda^6K^9$}
\nextgroupplot[
    xmode=log,
    xlabel={$K\lambda$}, ylabel={exact / asymptotic},
    title={\footnotesize low-band agreement},
    ymin=0.0, ymax=1.15,
]
\addplot[euvC, thick] table[x=Klam, y=ratio] {figures/data/transport_knee.dat};
\draw[gray, densely dotted] ({rel axis cs:0,0}|-{axis cs:0.02,1}) -- ({rel axis cs:1,0}|-{axis cs:0.02,1});
\end{groupplot}
\end{tikzpicture}
\caption{\textbf{Model-internal numerical check of the ninth-order transport-knee law.} Left: the exact finite-band conditional information \(\calI_{\lambda\mid\alpha,\beta}(K)/w\) (solid), computed as a backward-stable QR projection of \(g_\lambda\) off \(\operatorname{span}\{1,k^2\}\), against the closed-form law \((64/1225)\lambda^6K^9\) (dashed) of \cref{thm:k9}. The two coincide over more than four decades below the knee and separate only as \(K\lambda\to1\), where the neglected \(O(K^{11})\) curvature enters. Right: their ratio, which approaches unity as \(K\lambda\to0\); the fitted low-band log--log slope is \(\approx 8.98\), numerically confirming the predicted exponent \(9\). This figure is a consistency check for the implementation of the asymptotic calculation and the corrected prefactor \(I_{\lambda\lambda}=9w/\lambda^3(\cdots)\) of \cref{rem:correction}: a wrong power of \(\lambda\) would shift the curve vertically and break the ratio. It is generated from the same model assumptions and is not independent empirical evidence.}
\label{fig:transport-knee}
\end{figure}

\subsection{Finite-training penalty}

\cref{tab:finite-training} and \cref{fig:finite-training} quantify how plug-in templates from \(Q\) training releases fall below the ideal prediction, instantiating \cref{prop:bounds}(ii).

\begin{table}[htbp]
\centering
\small
\renewcommand{\arraystretch}{1.15}
\caption{\textbf{Finite-training penalty of plug-in templates.} At $L=109$, $K\lambda=1$, $\Delta\lambda/\lambda=10\%$ the adversary no longer knows the challenge spectra exactly but estimates each template mean from $Q$ training releases (the maximum-likelihood gamma-mean estimator) and scores with the plug-in likelihood. The table contrasts this realistic success with the ideal-template success on the same challenges; their gap is the price of finite profiling and is the empirical counterpart of the $O_p(\sqrt{L/Q})$ score perturbation in \cref{prop:bounds}. Success is over $2000$ Monte-Carlo challenges per row and the bracket is the Wilson $95\%$ interval.}
\label{tab:finite-training}
\setlength{\tabcolsep}{8pt}
\begin{tabular}{r c c c c}
\toprule
{$Q$ per template} & {plug-in success} & {plug-in 95\% CI} & {ideal success} & {penalty} \\
\midrule
5 & 0.527 & $[0.505,\,0.549]$ & 0.983 & 0.455 \\
10 & 0.497 & $[0.475,\,0.519]$ & 0.983 & 0.485 \\
20 & 0.495 & $[0.473,\,0.517]$ & 0.985 & 0.490 \\
50 & 0.645 & $[0.624,\,0.666]$ & 0.988 & 0.343 \\
100 & 0.693 & $[0.672,\,0.713]$ & 0.988 & 0.295 \\
500 & 0.850 & $[0.834,\,0.865]$ & 0.980 & 0.130 \\
\bottomrule
\end{tabular}
\par\vspace{3pt}
\begin{minipage}{0.92\linewidth}\footnotesize The penalty decays as the training budget grows, from nearly $0.5$ at $Q=5$ to about $0.13$ at $Q=500$. At $Q\le20$ the plug-in interval still covers $\tfrac12$, so a sparsely profiled adversary is statistically at chance, and the small non-monotonicity near $Q=5$ is Monte-Carlo scatter; only for $Q\gtrsim50$ does success rise decisively above chance, still far below the information-theoretic optimum the released channel would otherwise permit.\end{minipage}
\end{table}

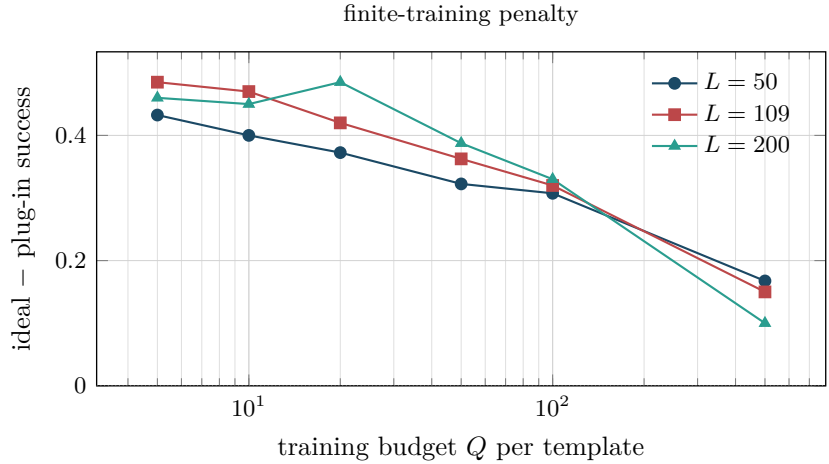
\begin{figure}[htbp]
\centering
\centering
\begin{tikzpicture}
\begin{axis}[
    euvaxis, width=0.68\linewidth, height=6.0cm,
    xmode=log, log basis x=10,
    xlabel={training budget $Q$ per template},
    ylabel={ideal $-$ plug-in success},
    ymin=0, legend pos=north east, legend cell align=left,
    title={\footnotesize finite-training penalty},
]
\addplot[euvA, mark=*, thick]      table[x=Q, y=L50]  {figures/data/finite_training.dat};
\addlegendentry{$L=50$}
\addplot[euvB, mark=square*, thick] table[x=Q, y=L109] {figures/data/finite_training.dat};
\addlegendentry{$L=109$}
\addplot[euvC, mark=triangle*, thick] table[x=Q, y=L200] {figures/data/finite_training.dat};
\addlegendentry{$L=200$}
\draw[gray, densely dotted] ({rel axis cs:0,0}|-{axis cs:5,0}) -- ({rel axis cs:1,0}|-{axis cs:5,0});
\end{axis}
\end{tikzpicture}
\caption{\textbf{The finite-training penalty closes as profiling improves.} Each curve fixes a challenge length \(L\) and plots the gap between ideal-template success and plug-in success against the per-template training budget \(Q\) (log scale). The penalty is large and \(L\)-dependent when \(Q\) is small (a sparsely profiled adversary is far from the channel optimum) and decays steadily toward zero as \(Q\) grows (up to Monte-Carlo scatter), in agreement with the \(O_p(\sqrt{L/Q})\) score perturbation of \cref{prop:bounds}(ii). Longer challenges (larger \(L\)) start from a worse penalty because the plug-in template error is amplified over more scored releases, but all curves collapse together once \(Q\) is large enough to estimate the templates well. Operationally, the curve tells a defender how much labeled history an attacker must accumulate before the released-channel exponent becomes the binding constraint.}
\label{fig:finite-training}
\end{figure}

\subsection{Correlated-bin covariance ablation}

\cref{tab:covariance-ablation} and \cref{fig:covariance} compare the diagonal-gamma exponent with an AR(1) log-residual covariance of the same per-bin variances: positive correlation reduces distinguishable information, so a diagonal audit over-states leakage. On a measured release the covariance must be estimated from repeats and shrunk when \(r\lesssim n\), following \eqref{eq:cov-kl} and the sample-cost note there; the ablation here instead uses a known AR(1) family to isolate the effect of correlation alone.

\begin{table}[htbp]
\centering
\small
\renewcommand{\arraystretch}{1.15}
\caption{\textbf{Correlated-bin ablation of the leakage exponent.} The protected mean difference is held fixed while only the noise model changes: the diagonal column uses $\Sigma=\operatorname{diag}(\psi_1(m))$ (independent gamma bins) and the AR(1) column uses a Toeplitz log-residual covariance $\Sigma_{ij}=\psi_1(m)\,\rho^{|i-j|}$ with the \emph{same} per-bin variances but correlation $\rho$. Both exponents come from the covariance-weighted local Chernoff $C=\tfrac18\,d^\top\Sigma^{-1}d$ of \eqref{eq:cov-kl}, isolating the effect of correlation alone.}
\label{tab:covariance-ablation}
\setlength{\tabcolsep}{8pt}
\begin{tabular}{c c c c}
\toprule
{$\rho$} & {$C$ (diagonal gamma)} & {$C$ (AR(1) bins)} & {ratio} \\
\midrule
0.0 & $2.36\times10^{-3}$ & $2.36\times10^{-3}$ & 1.000 \\
0.2 & $2.36\times10^{-3}$ & $1.62\times10^{-3}$ & 0.688 \\
0.4 & $2.36\times10^{-3}$ & $1.10\times10^{-3}$ & 0.467 \\
0.6 & $2.36\times10^{-3}$ & $7.19\times10^{-4}$ & 0.304 \\
0.8 & $2.36\times10^{-3}$ & $4.53\times10^{-4}$ & 0.192 \\
\bottomrule
\end{tabular}
\par\vspace{3pt}
\begin{minipage}{0.92\linewidth}\footnotesize Positive correlation pools redundant bins and shrinks the usable information: by $\rho=0.5$ the true exponent is already about $0.38$ of the diagonal value, so an independence-assuming audit \emph{over-states} leakage and should be read as an upper bound on real overlapped-Welch spectra.\end{minipage}
\end{table}

\begin{figure}[htbp]
\centering
\centering
\begin{tikzpicture}
\begin{axis}[
    euvaxis, width=0.66\linewidth, height=5.6cm,
    xlabel={bin correlation $\rho$ (AR(1) log-residuals)},
    ylabel={correlated / diagonal exponent},
    ymin=0, ymax=1.08, xmin=-0.02, xmax=0.92,
    title={\footnotesize diagonal gamma vs correlated-bin leakage},
]
\addplot[euvA, mark=*, thick] table[x=rho, y=ratio] {figures/data/covariance.dat};
\addplot[gray, densely dotted] coordinates {(-0.02,1) (0.92,1)};
\node[anchor=south east, font=\footnotesize, gray] at (axis cs:0.9,1.0) {independence ($=1$)};
\end{axis}
\end{tikzpicture}
\caption{\textbf{Bin correlation makes the diagonal audit conservative.} Ratio of the correlated-bin leakage exponent to the diagonal-gamma exponent as a function of the AR(1) log-residual correlation \(\rho\), with the protected mean difference and the per-bin variances held fixed. At \(\rho=0\) the ratio is one by construction; as \(\rho\) grows, neighbouring bins carry redundant information, the covariance-weighted Chernoff \(\tfrac18 d^\top\Sigma^{-1}d\) shrinks, and the ratio falls below one, reaching \(\approx0.38\) at \(\rho=0.5\) and continuing downward. The practical reading, \emph{restricted to the tested positive AR(1) covariance family}, is that the diagonal-gamma assumption stays conservative on these cases; real overlapped-Welch spectra call for substituting the measured covariance \(\widehat\Sigma\) for \(\operatorname{diag}(m_i)\).}
\label{fig:covariance}
\end{figure}
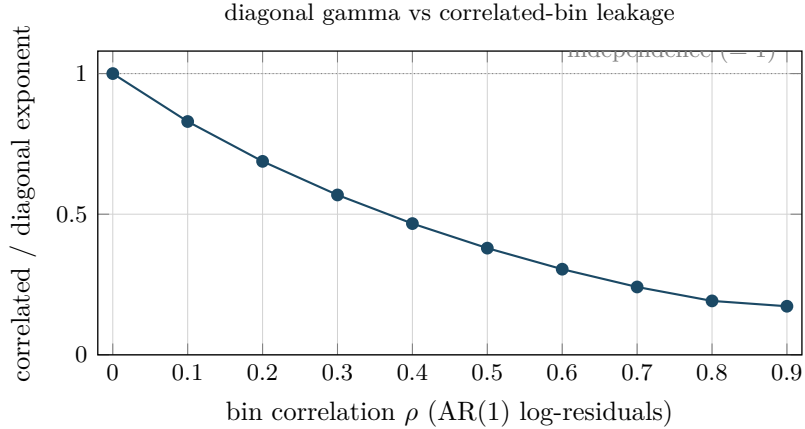

\subsection{Floor/blur/transport coupling map}

\cref{tab:coupling-diagnostics} and \cref{fig:coupling} map the floor-aware information and nuisance-block condition number over \((K\lambda,S_0/A)\); larger floor and band edge raise the condition number (\cref{prop:floor-coupling}).

\begin{table}[htbp]
\centering
\small
\renewcommand{\arraystretch}{1.15}
\caption{\textbf{Floor/blur/transport coupling diagnostics.} Representative cells of the $(K\lambda,\,S_0/A)$ grid giving the floor-aware conditional transport information $\mathcal{I}_{\lambda\mid\alpha,\beta,S_0}$ (the Schur complement of $\lambda$ after eliminating amplitude, blur, \emph{and} the metrology floor) together with the condition number of the $3\times3$ nuisance block. The condition number is a separability gauge: when it is large, the released band cannot cleanly distinguish a transport change from a joint amplitude/blur/floor adjustment (\cref{prop:floor-coupling}).}
\label{tab:coupling-diagnostics}
\setlength{\tabcolsep}{8pt}
\begin{tabular}{c c c c}
\toprule
{$K\lambda$} & {$S_0/A$} & {$\mathcal{I}_{\lambda\mid\alpha,\beta,S_0}$} & {cond(nuisance)} \\
\midrule
0.30 & 0.00 & $2.87\times10^{-10}$ & $1.53\times10^{8}$ \\
0.30 & 0.20 & $1.95\times10^{-10}$ & $1.54\times10^{8}$ \\
0.30 & 0.40 & $1.41\times10^{-10}$ & $1.54\times10^{8}$ \\
3.00 & 0.00 & $1.58\times10^{-1}$ & $1.51\times10^{3}$ \\
3.00 & 0.20 & $3.66\times10^{-2}$ & $4.72\times10^{2}$ \\
3.00 & 0.40 & $1.83\times10^{-2}$ & $4.05\times10^{2}$ \\
\bottomrule
\end{tabular}
\par\vspace{3pt}
\begin{minipage}{0.92\linewidth}\footnotesize Information rises steeply with band edge, but so does the conditioning: a floor-blind audit is trustworthy only in the well-conditioned corner (small $S_0/A$, modest $K\lambda$); elsewhere the floor tangent must be carried explicitly.\end{minipage}
\end{table}

\begin{figure}[htbp]
\centering
\input{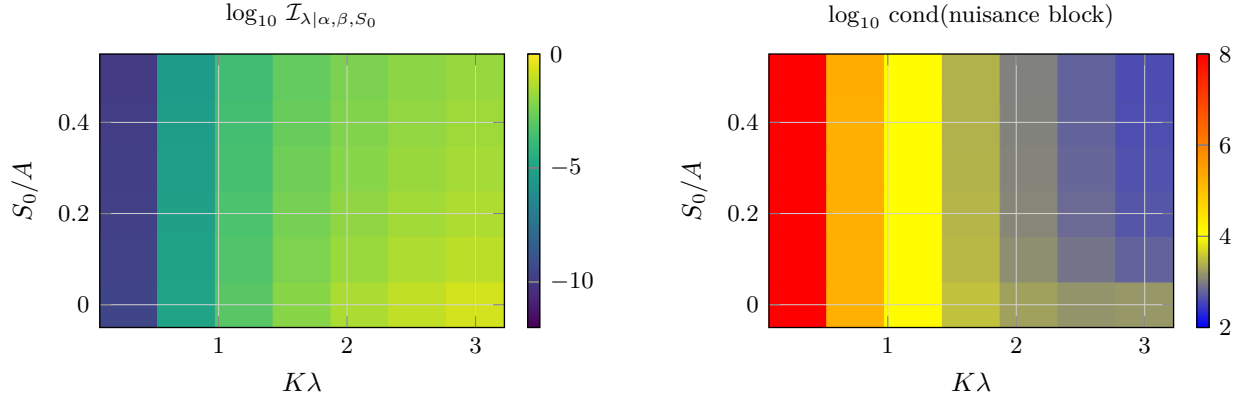}
\centering
\begin{tikzpicture}
\begin{groupplot}[
    group style={group size=2 by 1, horizontal sep=3.5cm},
    width=0.42\linewidth, height=5.2cm, euvaxis,
    enlargelimits=false, axis on top,
    xlabel={$K\lambda$}, ylabel={$S_0/A$},
    colorbar, colorbar style={width=0.16cm, font=\scriptsize},
    point meta min=-12,
]
\nextgroupplot[title={\footnotesize $\log_{10}\,\mathcal{I}_{\lambda\mid\alpha,\beta,S_0}$},
    colormap/viridis, point meta max=0]
\addplot[matrix plot*, mesh/cols=\couplingcols, point meta=explicit]
    table[x=Klam, y=rho, meta=val] {figures/data/coupling_info.dat};
\nextgroupplot[title={\footnotesize $\log_{10}\,\mathrm{cond}$(nuisance block)},
    colormap/hot, point meta min=2, point meta max=8]
\addplot[matrix plot*, mesh/cols=\couplingcols, point meta=explicit]
    table[x=Klam, y=rho, meta=val] {figures/data/coupling_cond.dat};
\end{groupplot}
\end{tikzpicture}
\caption{\textbf{Where a floor-blind audit can be trusted.} Heatmaps over the released band edge \(K\lambda\) (horizontal) and the relative metrology floor \(S_0/A\) (vertical). Left: \(\log_{10}\) of the floor-aware conditional transport information \(\calI_{\lambda\mid\alpha,\beta,S_0}\), the Schur complement of \(\lambda\) after eliminating amplitude, blur, \emph{and} floor; it rises steeply with band edge as the knee enters the window. Right: \(\log_{10}\) of the condition number of the \(3\times3\) nuisance block, a direct separability gauge. The two panels must be read together: information is largest exactly where the conditioning is worst (large \(K\lambda\), large \(S_0/A\)), so raw transport information there is partly an artifact of an ill-posed nuisance subtraction. A floor-blind audit (\cref{prop:floor-coupling}) is reliable only in the well-conditioned corner (small floor, moderate band), and elsewhere the floor tangent \(\partial_{S_0}f\) must be carried explicitly.}
\label{fig:coupling}
\end{figure}

\subsection{Leakage-utility release optimizer}

\cref{fig:release} traces the rank-one optimum of \cref{thm:minimax}: maximum retained utility of a public scalar statistic subject to a protected leakage budget, with the converse.

\begin{figure}[htbp]
\centering
\centering
\begin{tikzpicture}
\begin{axis}[
    euvaxis, width=0.66\linewidth, height=5.6cm,
    xmode=log, log basis x=10,
    xlabel={protected leakage budget $\varepsilon$},
    ylabel={retained utility $|\langle z,u\rangle|$},
    ymin=0, ymax=1.05,
    legend pos=north west, legend cell align=left,
    title={\footnotesize leakage--utility frontier ($\rho=-0.80$)},
]
\addplot[euvA, thick] table[x=eps, y=utility] {figures/data/release_frontier.dat};
\addlegendentry{optimal utility}
\addplot[euvB, dashed, thick] table[x=eps, y=utility] {figures/data/release_converse.dat};
\addlegendentry{converse bound}
\end{axis}
\end{tikzpicture}
\caption{\textbf{The minimum-leakage public statistic under a utility constraint.} For a protected tangent \(p\) (transport) and a utility tangent \(u\) (amplitude) with correlation \(\rho\), the curve gives the maximum retained utility \(|\inner{z}{u}|\) of an optimal rank-one release \(z\) subject to the protected leakage budget \(\inner{z}{p}^2/(2\tau^2)\le\eps\) (\cref{thm:minimax}); the dashed line is the matching converse, the least leakage any release achieving a given utility must incur. At small \(\eps\) the optimal release lives almost entirely in \(p^\perp\) and retains the component of \(u\) orthogonal to \(p\), namely \(\sqrt{1-\rho^2}\); as the budget grows the release is allowed to tilt toward \(u\) until, at \(2\tau^2\eps=\rho^2\), it aligns with \(u\) and retains full utility. The knee of this frontier is the operating point a release owner should target: it is the most useful public scalar whose protected exponent still respects the audit budget.}
\label{fig:release}
\end{figure}
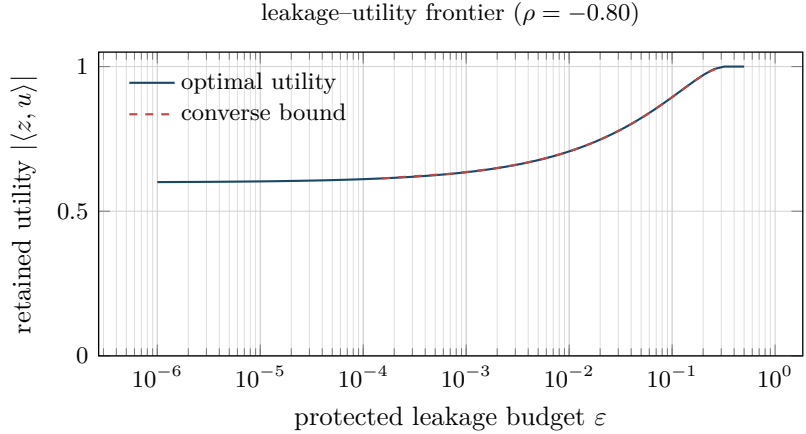

\section{Relation to prior work}

The framework is an instance of quantitative information flow \cite{smith2009,alvim2020}: leakage is measured by the statistical distinguishability a release induces between hidden states, here recipe coordinates rather than secret keys. The closest operational analogue is the profiled, or template, side-channel attack \cite{chari2003,choudary2013}, in which an adversary first estimates per-state templates from labeled traces and then attributes a challenge by likelihood. Our finite-training penalty (\cref{prop:bounds}(ii)) is exactly the template-estimation error of that setting; our covariance-weighted channel \eqref{eq:cov-kl} plays the role of the pooled trace covariance, shrinkage-regularized when repeats are scarce \cite{ledoit2004}; and our positive-margin mismatch bound (\cref{prop:bounds}(iii)) is the profiling-model-mismatch problem made quantitative. Information-theoretic side-channel evaluation \cite{standaert2009} and differential power analysis \cite{kocher1999} supply the leakage-estimation viewpoint, classical testing supplies the KL/TV/Chernoff machinery \cite{cover2006}, and the additive release accounting parallels differential privacy \cite{dwork2006} and the broader question of what a released statistic or model discloses about its training data \cite{shokri2017}, i.e.\ the statistical-disclosure-control problem of bounding what a public release reveals \cite{dinur2003}. We borrow the indistinguishability-experiment \emph{language} of \cite{goldwasser1984}; the contribution is a quantitative leakage audit for a metrology-release channel.

On the lithography side, chemical amplification underlies acid generation and deprotection \cite{ito1983,reichmanis1991,wallraff1999}; Houle and coauthors measured coupled catalysis/diffusion \cite{houle2000,houle2004}; acid diffusion is a performance variable \cite{vansteenwinckel2005}. Gallatin modeled blur and LER \cite{gallatin2005,gallatin2008}; Naulleau and De Bisschop framed stochastic patterning and failures \cite{naulleau2011,deBisschop2017}; Mack treated shot noise and metrics \cite{mack2018shot,mack2019metrics}. PSD analysis is standard \cite{constantoudis2004,lorusso2018,cutler2021,palasantzas1993,wallow2008}.

\section{Scope and physical limitations of the EUV instantiation}
\label{sec:limitations}

The audit calculus (the gamma and covariance-weighted channels, the KL/Chernoff and protected-bit bounds, the finite-band geometry, the release-map optimizer, and the protocol of \cref{sec:protocol}) is the contribution, and it is model-agnostic: it applies to any declared released-PSD channel. The screened spectrum \eqref{eq:screened-psd} instantiates that calculus for the EUV edge regime.

\emph{The screened form models the acid-diffusion edge regime.} The factor \((1+\lambda^2k^2)^{-3/2}\) is derived in \cref{app:green} from a screened reaction--diffusion edge field and carries a transport \emph{knee} near \(k\sim1/\lambda\) characteristic of that regime. The residual gate~P3 of \cref{sec:protocol} confirms the form on each release before any exponent is quoted, so the audit applies it where it fits and defers elsewhere.

\emph{The protected coordinate is effective.} A single released spectrum identifies only the combination \(\mu=\lambda^{-2}=(k_qq_0+k_{\rm loss})/D_H\) (\cref{sec:transport}); separating quencher loading \(q_0\) from the loss rate requires a controlled sweep. Throughout, \(\lambda\) is an \emph{effective} transport coordinate, and the leakage statements concern \(\lambda\) (equivalently \(\mu\)), not an individually resolved chemical rate.

\emph{The numbers are model-conditioned, and the calculus is ready for measured deployment.} The reconstructed-PSD fit is a pipeline smoke test, the published 18-nm figures set the scale, and the gamma-channel audits are model-internal Monte Carlo that exercise the calculus end to end. The next step is to run the protocol of \cref{sec:protocol} on repeated measured releases, estimating the effective degrees of freedom and covariance, testing the gamma and screened assumptions, and reporting the safe band, which is the deployment the calculus is built for.

\section{Conclusion}

A released roughness spectrum is a statistical transcript of a recipe: once release map, estimator, and budgets are fixed, leakage is a likelihood exponent, exact for gamma PSDs and the Fisher norm of the released tangent locally. For acid-quencher transport this exponent obeys \(\calI_{\lambda\mid\alpha,\beta}(K)=\tfrac{64}{1225}w\lambda^6K^9+O(w\lambda^8K^{11})\), giving a closed-form safe-band edge and attribution thresholds; the release-map geometry places protected tangents in the nullspace and keeps utility inside it; and a fixed-seed package together with the deployment protocol of \cref{sec:protocol} turns the calculus into numbers a release owner can act on. We instantiate everything on screened EUV roughness spectra as a model-conditioned case study, and running the same protocol on repeated measured releases is the natural next step, which it is written to support unchanged.

\appendix
\renewcommand{\thesection}{\Alph{section}}
\setcounter{section}{0}

\section*{Appendices: worked derivations}
The appendices collect the full computations behind the results quoted in the main text. They are deliberately explicit (each integral is carried out, each Gram matrix is inverted, each optimization is solved by hand), so that the package's unit tests can be read against the algebra line by line.

\section{Screened Green-function reduction}
\label{app:green}

We derive the screened factor \((1+\lambda^2k^2)^{-3/2}\) of \eqref{eq:screened-psd} from a two-dimensional reaction--diffusion field. Let \(x\) run along the edge and \(y\) normal to it. Linearized acid-concentration fluctuations \(a(x,y)\) obey a screened diffusion equation driven by a white exposure source \(\eta\),
\begin{equation}
\big(-D_H\nabla^2+D_H\lambda^{-2}\big)\,a(x,y)=\eta(x,y),
\qquad \E[\eta(\mathbf r)\eta(\mathbf r')]=\Gamma\,\delta(\mathbf r-\mathbf r'),
\label{eq:appA-pde}
\end{equation}
so the Green function \(G\) solving \((-D_H\nabla^2+D_H\lambda^{-2})G=\delta\) has the two-dimensional Fourier transform
\begin{equation}
\widehat G(k,q)=\frac{1}{D_H\,(k^2+q^2+\lambda^{-2})},
\label{eq:appA-green}
\end{equation}
with \(k,q\) conjugate to \(x,y\). The field PSD is \(S_a(k,q)=\Gamma\,|\widehat G(k,q)|^2\). The edge is the level set of \(a\); to first order its displacement is \(u(x)=-a(x,0)/G_e\) with \(G_e\) the mean threshold gradient, so the one-dimensional edge PSD is the marginal of \(S_a\) over the unresolved normal wavenumber:
\begin{equation}
S_u(k)=\frac{1}{G_e^2}\int_{-\infty}^{\infty}\frac{\dd q}{2\pi}\,S_a(k,q)
=\frac{\Gamma}{G_e^2D_H^2}\int_{-\infty}^{\infty}\frac{\dd q}{2\pi}\,\frac{1}{(k^2+q^2+\lambda^{-2})^2}.
\label{eq:appA-marginal}
\end{equation}
Write \(\kappa^2=k^2+\lambda^{-2}\). The standard integral
\begin{equation}
\int_{-\infty}^{\infty}\frac{\dd q}{(q^2+\kappa^2)^2}=\frac{\pi}{2\kappa^3}
\label{eq:appA-int}
\end{equation}
follows by residues (a double pole at \(q=i\kappa\)) or by differentiating \(\int(q^2+\kappa^2)^{-1}\dd q=\pi/\kappa\) with respect to \(\kappa^2\). Hence
\begin{equation}
S_u(k)=\frac{\Gamma}{G_e^2D_H^2}\cdot\frac{1}{2\pi}\cdot\frac{\pi}{2(k^2+\lambda^{-2})^{3/2}}
=\frac{\Gamma}{4G_e^2D_H^2}\,(k^2+\lambda^{-2})^{-3/2}.
\label{eq:appA-pre}
\end{equation}
Factoring \(\lambda^{-2}\) out of the bracket, \((k^2+\lambda^{-2})^{-3/2}=\lambda^{3}(1+\lambda^2k^2)^{-3/2}\), gives
\begin{equation}
S_u(k)=A_0\,(1+\lambda^2k^2)^{-3/2},\qquad A_0=\frac{\Gamma\lambda^3}{4G_e^2D_H^2}.
\label{eq:appA-final}
\end{equation}
Multiplying by the aggregate Gaussian blur \(e^{-\beta k^2}\) (optical, secondary-electron, development, and metrology) and adding a white metrology floor \(S_0\) yields \eqref{eq:screened-psd}. Two remarks. First, \(A_0\propto\lambda^3\): a longer screening length both narrows the spectrum and raises its plateau, which is why amplitude and \(\lambda\) are entangled until the curvature term is isolated. Second, the exponent \(3/2\) is specific to the 2-D-to-1-D marginalization in \eqref{eq:appA-marginal}; a different roughness exponent would replace the screened factor and rescale every Fisher entry below.

\section{Exact finite-band Fisher entries}
\label{app:fisher}

We derive \eqref{eq:fisher-basic}--\eqref{eq:fisher-transport}. Throughout \(w\) is the flat averaging density, \(X=K\lambda\), and the inner product is \(\inner{g}{h}=w\int_0^K g(k)h(k)\,\dd k\). The floor-free log-spectrum is \(f=\alpha-\beta k^2-\tfrac32\log(1+\lambda^2k^2)\), giving the tangents
\begin{equation}
\partial_\alpha f=1,\qquad \partial_\beta f=-k^2,\qquad
\partial_\lambda f=g_\lambda(k)=-\frac{3\lambda k^2}{1+\lambda^2k^2}.
\label{eq:appB-tangents}
\end{equation}

\paragraph{Amplitude--blur block.} These are elementary monomial integrals:
\begin{align}
I_{\alpha\alpha}&=w\!\int_0^K 1\,\dd k=wK, &
I_{\alpha\beta}&=w\!\int_0^K(-k^2)\,\dd k=-\tfrac{wK^3}{3}, &
I_{\beta\beta}&=w\!\int_0^K k^4\,\dd k=\tfrac{wK^5}{5}.
\label{eq:appB-block}
\end{align}

\paragraph{Cross terms.} Substitute \(u=\lambda k\), \(\dd k=\dd u/\lambda\), upper limit \(X=K\lambda\). For \(I_{\alpha\lambda}\),
\begin{equation}
I_{\alpha\lambda}=w\!\int_0^K\!\!\Big(\!-\frac{3\lambda k^2}{1+\lambda^2k^2}\Big)\dd k
=-3w\lambda\!\int_0^K\!\frac{k^2}{1+\lambda^2k^2}\dd k
=-\frac{3w}{\lambda^2}\!\int_0^X\!\frac{u^2}{1+u^2}\dd u.
\label{eq:appB-al1}
\end{equation}
Since \(u^2/(1+u^2)=1-1/(1+u^2)\) and \(\int_0^X(1+u^2)^{-1}\dd u=\arctan X\),
\begin{equation}
I_{\alpha\lambda}=-\frac{3w}{\lambda^2}\big(X-\arctan X\big).
\label{eq:appB-al2}
\end{equation}
For \(I_{\beta\lambda}\), the integrand is \((-k^2)\,g_\lambda=3\lambda k^4/(1+\lambda^2k^2)\):
\begin{equation}
I_{\beta\lambda}=3w\lambda\!\int_0^K\!\frac{k^4}{1+\lambda^2k^2}\dd k
=\frac{3w}{\lambda^4}\!\int_0^X\!\frac{u^4}{1+u^2}\dd u.
\label{eq:appB-bl1}
\end{equation}
Polynomial division gives \(u^4/(1+u^2)=u^2-1+1/(1+u^2)\), so \(\int_0^X=\tfrac{X^3}{3}-X+\arctan X\) and
\begin{equation}
I_{\beta\lambda}=\frac{3w}{\lambda^4}\Big(\frac{X^3}{3}-X+\arctan X\Big).
\label{eq:appB-bl2}
\end{equation}

\paragraph{Transport diagonal (the corrected entry).} Here \(g_\lambda^2=9\lambda^2k^4/(1+\lambda^2k^2)^2\):
\begin{equation}
I_{\lambda\lambda}=9w\lambda^2\!\int_0^K\!\frac{k^4}{(1+\lambda^2k^2)^2}\dd k
=\frac{9w}{\lambda^3}\!\int_0^X\!\frac{u^4}{(1+u^2)^2}\dd u.
\label{eq:appB-ll1}
\end{equation}
The explicit \(\lambda\) scaling is \(9\lambda^2\cdot\lambda^{-5}=9\lambda^{-3}\): \emph{five} inverse powers from \(\dd k\) and \(u^4\), two positive from the prefactor. To evaluate the \(u\)-integral, write \(u^4=(u^2+1)^2-2(u^2+1)+1\), so
\begin{equation}
\frac{u^4}{(1+u^2)^2}=1-\frac{2}{1+u^2}+\frac{1}{(1+u^2)^2}.
\label{eq:appB-split}
\end{equation}
Using \(\int_0^X(1+u^2)^{-1}\dd u=\arctan X\) and \(\int_0^X(1+u^2)^{-2}\dd u=\tfrac12\big(\arctan X+\tfrac{X}{1+X^2}\big)\),
\begin{equation}
\int_0^X\frac{u^4}{(1+u^2)^2}\dd u
=X-2\arctan X+\tfrac12\arctan X+\tfrac{X}{2(1+X^2)}
=X-\tfrac32\arctan X+\tfrac{X}{2(1+X^2)},
\label{eq:appB-ll2}
\end{equation}
which gives \eqref{eq:fisher-transport}. As a check, expand for small \(X\): \(\arctan X=X-\tfrac{X^3}{3}+\tfrac{X^5}{5}-\cdots\) and \(\tfrac{X}{2(1+X^2)}=\tfrac{X}{2}(1-X^2+X^4-\cdots)\). Collecting orders, the \(X\) and \(X^3\) terms cancel and the leading survivor is \(\tfrac15 X^5\), so \(I_{\lambda\lambda}\to(9w/\lambda^3)\cdot\tfrac15 X^5=\tfrac95 w\lambda^2K^5\), the value pinned by \texttt{test\_fisher\_transport\_knee}. An entry \(9w/\lambda\) would violate this limit (which must equal \(\int_0^K w(3\lambda k^2)^2\dd k=\tfrac95 w\lambda^2K^5\)), and is the typo corrected in \cref{rem:correction}.

\section{The ninth-order residual computation}
\label{app:k9}

We prove \eqref{eq:k9}. The conditional information is the squared \(L^2_w\)-norm of \(g_\lambda\) after removing its projection onto \(\calB=\operatorname{span}\{1,k^2\}\). Expand \(g_\lambda\) at small \(\lambda\):
\begin{equation}
g_\lambda(k)=-3\lambda k^2\big(1+\lambda^2k^2\big)^{-1}
=-3\lambda k^2+3\lambda^3k^4-3\lambda^5k^6+\cdots
\label{eq:appC-expand}
\end{equation}
The constant-in-shape pieces \(-3\lambda k^2\) lie \emph{entirely} in \(\calB\), so they are annihilated by \(I-\Pi_\calB\). The first surviving term is \(3\lambda^3k^4\); thus
\begin{equation}
\calI_{\lambda\mid\alpha,\beta}(K)=(3\lambda^3)^2\,\big\|(I-\Pi_\calB)\,k^4\big\|_{w,K}^2+O(\lambda^8K^{11}),
\label{eq:appC-reduce}
\end{equation}
and it remains to compute the residual of \(k^4\) against \(\{1,k^2\}\) on \([0,K]\). Seek \(k^4=c_0+c_2k^2+r_K(k)\) with \(r_K\perp_w\{1,k^2\}\). With the moment shorthand \(\mu_j=\int_0^K k^j\,\dd k=K^{j+1}/(j+1)\), the normal equations are
\begin{equation}
\begin{pmatrix}\mu_0&\mu_2\\ \mu_2&\mu_4\end{pmatrix}\!\begin{pmatrix}c_0\\ c_2\end{pmatrix}
=\begin{pmatrix}\mu_4\\ \mu_6\end{pmatrix},
\qquad
\begin{pmatrix}\mu_0&\mu_2\\ \mu_2&\mu_4\end{pmatrix}
=\begin{pmatrix}K&\tfrac{K^3}{3}\\[2pt]\tfrac{K^3}{3}&\tfrac{K^5}{5}\end{pmatrix},
\quad
\begin{pmatrix}\mu_4\\ \mu_6\end{pmatrix}
=\begin{pmatrix}\tfrac{K^5}{5}\\[2pt]\tfrac{K^7}{7}\end{pmatrix}.
\label{eq:appC-normal}
\end{equation}
The Gram determinant is \(\mu_0\mu_4-\mu_2^2=\tfrac{K^6}{5}-\tfrac{K^6}{9}=\tfrac{4K^6}{45}\). Cramer's rule gives
\begin{equation}
c_0=\frac{\mu_4\mu_4-\mu_2\mu_6}{\det}
=\frac{\tfrac{K^{10}}{25}-\tfrac{K^{10}}{21}}{\tfrac{4K^6}{45}}
=\frac{-\tfrac{4}{525}K^{10}}{\tfrac{4}{45}K^6}=-\frac{3}{35}K^4,
\label{eq:appC-c0}
\end{equation}
\begin{equation}
c_2=\frac{\mu_0\mu_6-\mu_2\mu_4}{\det}
=\frac{\tfrac{K^8}{7}-\tfrac{K^8}{15}}{\tfrac{4K^6}{45}}
=\frac{\tfrac{8}{105}K^8}{\tfrac{4}{45}K^6}=\frac{6}{7}K^2,
\label{eq:appC-c2}
\end{equation}
so \(k^4=-\tfrac{3}{35}K^4+\tfrac67K^2k^2+r_K(k)\), matching the residual quoted in the proof of \cref{thm:k9}. By orthogonality the residual norm is
\begin{equation}
\|r_K\|^2=\int_0^K k^8\dd k-c_0\!\int_0^K k^4\dd k-c_2\!\int_0^K k^6\dd k
=\frac{K^9}{9}-c_0\frac{K^5}{5}-c_2\frac{K^7}{7}.
\label{eq:appC-resid}
\end{equation}
Substituting \(c_0,c_2\),
\begin{equation}
\|r_K\|^2=\Big(\frac19+\frac{3}{35}\cdot\frac15-\frac67\cdot\frac17\Big)K^9
=\Big(\frac{1}{9}+\frac{3}{175}-\frac{6}{49}\Big)K^9.
\label{eq:appC-collect}
\end{equation}
Over the common denominator \(11025=9\cdot175\cdot7/(\gcd)\) (explicitly \(11025=3^2\cdot5^2\cdot7^2\)): \(\tfrac19=\tfrac{1225}{11025}\), \(\tfrac{3}{175}=\tfrac{189}{11025}\), \(\tfrac{6}{49}=\tfrac{1350}{11025}\), whence
\begin{equation}
\|r_K\|^2=\frac{1225+189-1350}{11025}K^9=\frac{64}{11025}K^9.
\label{eq:appC-64}
\end{equation}
Finally, restoring the prefactor and the weight,
\begin{equation}
\calI_{\lambda\mid\alpha,\beta}(K)=(3\lambda^3)^2\,w\,\frac{64}{11025}K^9
=\frac{9\cdot64}{11025}\,w\lambda^6K^9=\frac{64}{1225}\,w\lambda^6K^9,
\label{eq:appC-done}
\end{equation}
using \(9/11025=1/1225\). This is \eqref{eq:k9}. The next term comes from the \(-3\lambda^5k^6\) piece of \eqref{eq:appC-expand} together with the cross term against the residual of \(k^4\); both scale as \(\lambda^8K^{11}\), giving the stated remainder. The quencher form of \cref{cor:qbit} follows by the chain rule \(\delta\lambda=-\tfrac12\lambda^3(k_q/D_H)\delta q_0\), squaring the prefactor: \((\tfrac12\lambda^3 k_q/D_H)^2\cdot\tfrac{64}{1225}w\lambda^6K^9=\tfrac{16}{1225}w(k_q^2/D_H^2)\lambda^{12}K^9\).

\section{Gamma-channel divergences}
\label{app:gamma}

Let \(\widehat S\sim\mathrm{Gamma}(m,S/m)\) have density \(p_S(y)=\frac{(m/S)^m}{\Gamma(m)}y^{m-1}e^{-my/S}\). For two means \(S,S'\) write \(r=S'/S\).

\paragraph{KL.} With \(\E_S[\widehat S]=S\) and \(\E_S[\log\widehat S]=\psi(m)+\log(S/m)\),
\begin{equation}
\KL(p_S\Vert p_{S'})=\E_S\Big[\log\frac{p_S}{p_{S'}}\Big]
=\E_S\Big[m\log\frac{S'}{S}+\Big(\frac{m}{S'}-\frac{m}{S}\Big)\widehat S\Big]
=m\Big[\log r+\frac1r-1\Big],
\label{eq:appD-kl}
\end{equation}
since the \(y^{m-1}\) and \(\Gamma(m)\) factors are shared and cancel, and \(\E_S[\widehat S](m/S'-m/S)=m(S/S'-1)=m(1/r-1)\). Summing over independent bins gives \eqref{eq:gamma-kl}. Note \(\log r+1/r-1\ge0\) with equality iff \(r=1\), so KL is nonnegative and vanishes only for identical spectra.

\paragraph{Chernoff.} The \(s\)-tilted (R\'enyi) integral is, for \(0<s<1\),
\begin{equation}
\int_0^\infty p_S(y)^{1-s}p_{S'}(y)^{s}\,\dd y
=\frac{(m/S)^{m(1-s)}(m/S')^{ms}}{\Gamma(m)}\int_0^\infty y^{m-1}e^{-my[(1-s)/S+s/S']}\dd y.
\label{eq:appD-tilt}
\end{equation}
The remaining integral is \(\Gamma(m)\,[m((1-s)/S+s/S')]^{-m}\), so
\begin{equation}
-\log\!\int p_S^{1-s}p_{S'}^{s}
=m\Big[(1-s)\log S+s\log S'+\log\!\Big(\frac{1-s}{S}+\frac{s}{S'}\Big)\Big],
\label{eq:appD-cher}
\end{equation}
the bracketed summand of \eqref{eq:chernoff}; the Chernoff information is its supremum over \(s\in[0,1]\), additive over bins. The objective is concave in \(s\) and vanishes at \(s\in\{0,1\}\); the package maximizes it by a bounded Brent search with a grid fallback (\texttt{chernoff.py}). At \(s=\tfrac12\) one obtains the Bhattacharyya exponent \(m[\tfrac12\log SS'+\log(\tfrac12(1/S+1/S'))]\) used as a cross-check.

\paragraph{Local limit.} For \(S'=S(1+\delta)\) with small \(\delta\), \eqref{eq:appD-kl} expands as \(\KL=\tfrac{m}{2}\delta^2+O(\delta^3)\) and the Chernoff supremum is attained near \(s=\tfrac12\) with value \(\tfrac{m}{8}\delta^2+O(\delta^3)\); hence \(C\approx\KL/4\) locally, the relation used for the covariance-weighted channel and verified in \texttt{test\_chernoff}.

\section{Floor perturbation constant}
\label{app:floor}

We make \cref{prop:floor-coupling} quantitative. Let \(g_\lambda\) be the floor-free transport tangent and \(g_\lambda^F=M_rg_\lambda\), where \(M_r\) is multiplication by \(r(k)=(1+\rho(k))^{-1}\) and \(\rho(k)=S_0/P(k)\le\rho_K\le\tau\). Since \(\|I-M_r\|=\sup_k|1-r(k)|=\sup_k\rho/(1+\rho)\le\tau\), \(M_r\) is a bounded perturbation of the identity. Let \(\Pi\) and \(\Pi_r\) be the orthogonal projections onto \(\calN=\operatorname{span}\{1,k^2\}\) and \(\calN_r=\operatorname{span}\{r,rk^2\}\). The Gram matrix of \(\calN\) is nonsingular on every finite band \(K>0\), so the projection map is locally Lipschitz in the two perturbed basis functions \(r\) and \(rk^2\): \(\|\Pi_r-\Pi\|\le \kappa_{\rm gap}\tau\), where \(\kappa_{\rm gap}\) is controlled by the inverse singular-value gap of the nuisance Gram matrix. Therefore the fixed-floor residual obeys
\begin{equation}
\big|\calI_{\lambda\mid\alpha,\beta}^{(r)}-\calI_{\lambda\mid\alpha,\beta}\big|
=\big|\,\|(I-\Pi_r)M_rg_\lambda\|^2-\|(I-\Pi)g_\lambda\|^2\big|
\le C_K\,\tau\,\|g_\lambda\|_{w,K}^2,
\label{eq:appE-bound}
\end{equation}
with \(C_K=2(1+\kappa_{\rm gap})+O(\tau)\). If the metrology floor is also profiled, the nuisance space is enlarged to \(\calN_F=\calN_r+\operatorname{span}\{\partial_{S_0}f\}\), so the residual norm can only decrease. The exact decrease is the squared projection of \((I-\Pi_r)g_\lambda^F\) onto the one residual floor direction \((I-\Pi_r)\partial_{S_0}f\), as written in the proof of \cref{prop:floor-coupling}. The subband statement is immediate: on a band where \(r(k)\le r_*\), every non-floor recipe tangent is multiplied by \(r\) and every raw Fisher integrand (a product of two such tangents) by \(r^2\le r_*^2\), so that subband contributes at most \(r_*^2\) of its floor-free value. The nuisance-block condition number tabulated in \cref{tab:coupling-diagnostics} and mapped in \cref{fig:coupling} is the numerical version of \(\kappa_{\rm gap}\); this is why large conditioning and untrustworthy floor-blind audits coincide.

\section{Release-map optimizers}
\label{app:optimizers}

\paragraph{Zero-leakage utility projection (\cref{thm:projection-synthesis}).} A deterministic projection release onto a \(d\)-dimensional subspace \(\calZ\) with isotropic noise has released Fisher \(\propto \Pi_\calZ G\); by the nullspace condition of \cref{thm:projection-synthesis} it has zero protected exponent iff \(\Pi_\calZ p=0\) for every protected \(p\), i.e.\ \(\calZ\subseteq\calP^\perp\). Subject to that, the retained utility energy is \(\sum_{a}\|\Pi_\calZ u_a\|^2=\Tr(\Pi_\calZ\Pi_\calU\Pi_\calZ)=\Tr(\Pi_\calZ\,\Pi_{\calP^\perp}\Pi_\calU\Pi_{\calP^\perp}\,\Pi_\calZ)\), the last equality because \(\Pi_\calZ=\Pi_{\calP^\perp}\Pi_\calZ\) on \(\calP^\perp\). Maximizing a trace \(\Tr(\Pi_\calZ A)\) over rank-\(d\) projections with \(A=\Pi_{\calP^\perp}\Pi_\calU\Pi_{\calP^\perp}\succeq0\) is Ky Fan's problem: the optimum is the span of the top \(d\) eigenvectors of \(A\), with value \(\sum_{i\le d}\sigma_i(A)\). This is exactly \texttt{release\_maps.zero\_leakage\_utility\_projection}, which forms \(A\) in the \(\sqrt{w}\)-weighted basis and returns its leading eigenspace; the unit test checks \(\Pi_\calP\calZ=0\) to machine precision.

\paragraph{Rank-one optimizer (\cref{thm:minimax}).} Only \(\operatorname{span}\{p,u\}\) matters, so write \(u=\rho p+\sqrt{1-\rho^2}\,q\) with \(q\perp p\), \(\|q\|=1\), and a candidate \(z=ap+bq\) with \(a^2+b^2\le1\). The leakage constraint is \(\inner{z}{p}^2/(2\tau^2)=a^2/(2\tau^2)\le\eps\), i.e.\ \(|a|\le A:=\sqrt{2\tau^2\eps}\). The utility is \(|\inner{z}{u}|=|a\rho+b\sqrt{1-\rho^2}|\). For fixed \(a\), \(b\) is maximized by saturating the norm, \(b=\sqrt{1-a^2}\), giving \(\phi(a)=a\rho+\sqrt{1-a^2}\sqrt{1-\rho^2}\) (signs aligned). Then \(\phi'(a)=\rho-\tfrac{a}{\sqrt{1-a^2}}\sqrt{1-\rho^2}\) vanishes at \(a^\star=\rho\) (where \(\phi=1\)), and \(\phi\) increases on \([0,\rho]\). Two regimes follow. If the budget is loose, \(A\ge|\rho|\), the unconstrained optimum \(a=\rho\) is feasible and the maximum utility is \(1\), attained at \(z=u\). If the budget binds, \(A<|\rho|\), then \(\phi\) is still increasing at \(a=A\), so the optimum saturates the constraint, \(|a|=A\), and
\begin{equation}
\max|\inner{z}{u}|=A\,|\rho|+\sqrt{1-A^2}\,\sqrt{1-\rho^2}
=\sqrt{2\tau^2\eps}\,|\rho|+\sqrt{1-2\tau^2\eps}\,\sqrt{1-\rho^2},
\label{eq:appF-util}
\end{equation}
which is \cref{thm:minimax} and the curve plotted in \cref{fig:release}. The scalar converse used for the dashed line inverts the same relation: a release achieving utility \(\gamma\) needs \(|a|\ge\) the smaller root of \(\phi(a)=\gamma\), i.e.\ protected exponent at least \(\tfrac{1}{2\tau^2}\big(\max\{0,\rho\gamma-\sqrt{(1-\rho^2)(1-\gamma^2)}\}\big)^2\), implemented as \texttt{release\_maps.leakage\_utility\_converse}.

\begin{thebibliography}{99}
\bibitem{ito1983} H. Ito and C. G. Willson, ``Chemical amplification in the design of dry developing resist materials,'' \emph{Polymer Eng.\ Sci.}, 23(18):1012--1018, 1983.
\bibitem{reichmanis1991} E. Reichmanis et al., ``Chemical amplification mechanisms for microlithography,'' \emph{Chem.\ Mater.}, 3(3):394--407, 1991.
\bibitem{wallraff1999} G. M. Wallraff and W. D. Hinsberg, ``Lithographic imaging techniques for the formation of nanoscopic features,'' \emph{Chem.\ Rev.}, 99(7):1801--1822, 1999.
\bibitem{houle2000} F. A. Houle et al., ``Determination of coupled acid catalysis-diffusion processes in a positive-tone chemically amplified photoresist,'' \emph{J.\ Vac.\ Sci.\ Technol.\ B}, 18(4):1874--1885, 2000.
\bibitem{houle2004} F. A. Houle et al., ``Acid-base reactions in a positive-tone chemically amplified photoresist,'' \emph{J.\ Vac.\ Sci.\ Technol.\ B}, 22(2):747--754, 2004.
\bibitem{vansteenwinckel2005} D. Van Steenwinckel et al., ``Lithographic importance of acid diffusion in chemically amplified resists,'' \emph{Proc.\ SPIE}, 5753:269--280, 2005.
\bibitem{fallica2016} R. Fallica et al., ``Dynamic absorption coefficients of CARs and non-CARs at EUV,'' \emph{J.\ Micro/Nanolith.\ MEMS MOEMS}, 15(3):033506, 2016.
\bibitem{welch1967} P. D. Welch, ``The use of fast Fourier transform for the estimation of power spectra,'' \emph{IEEE Trans.\ Audio Electroacoust.}, 15(2):70--73, 1967.
\bibitem{gallatin2005} G. M. Gallatin, ``Resist blur and line edge roughness,'' \emph{Proc.\ SPIE}, 5754:38--52, 2005.
\bibitem{gallatin2008} G. M. Gallatin, ``Resolution, LER, and sensitivity limitations of photoresist,'' \emph{Proc.\ SPIE}, 6921:69211E, 2008.
\bibitem{wallow2008} T. Wallow et al., ``Evaluation of EUV resist materials for use at the 32 nm half-pitch node,'' \emph{Proc.\ SPIE}, 6921:69211F, 2008.
\bibitem{naulleau2011} P. P. Naulleau et al., ``Critical challenges for EUV resist materials,'' \emph{Proc.\ SPIE}, 7972:797202, 2011.
\bibitem{naulleau2011metrology} P. P. Naulleau and B. McClinton, ``LER Metrology: Can we trust the numbers?,'' EUV Lithography Workshop P31, 2011.
\bibitem{brunner2017} T. A. Brunner et al., ``Line-edge roughness performance targets for EUV lithography,'' \emph{Proc.\ SPIE}, 10143:101430E, 2017.
\bibitem{deBisschop2017} P. De Bisschop, ``Stochastic effects in EUV lithography,'' \emph{J.\ Micro/Nanolith.\ MEMS MOEMS}, 16(4):041013, 2017.
\bibitem{mack2018shot} C. A. Mack, ``Shot noise: a 100-year history, with applications to lithography,'' \emph{J.\ Micro/Nanolith.\ MEMS MOEMS}, 17(4):041002, 2018.
\bibitem{mack2019metrics} C. A. Mack, ``Metrics for stochastic scaling in EUV lithography,'' \emph{Proc.\ SPIE}, 10957:109570F, 2019.
\bibitem{constantoudis2004} V. Constantoudis et al., ``Line edge roughness and critical dimension variation,'' \emph{J.\ Vac.\ Sci.\ Technol.\ B}, 22(4):1974--1981, 2004.
\bibitem{lorusso2018} G. F. Lorusso et al., ``The imec roughness protocol,'' \emph{Proc.\ SPIE}, 10585:105850D, 2018.
\bibitem{cutler2021} C. Cutler et al., ``Pattern roughness analysis using power spectral density,'' \emph{J.\ Micro/Nanopattern.\ Mater.\ Metrol.}, 20(1):010901, 2021.
\bibitem{palasantzas1993} G. Palasantzas, ``Roughness spectrum and surface width of self-affine fractal surfaces via the K-correlation model,'' \emph{Phys.\ Rev.\ B}, 48(19):14472--14478, 1993.
\bibitem{kocher1999} P. Kocher, J. Jaffe, and B. Jun, ``Differential power analysis,'' in \emph{CRYPTO '99}, LNCS 1666, pp.\ 388--397, Springer, 1999.
\bibitem{chari2003} S. Chari, J. R. Rao, and P. Rohatgi, ``Template attacks,'' in \emph{CHES 2002}, LNCS 2523, pp.\ 13--28, Springer, 2003.
\bibitem{goldwasser1984} S. Goldwasser and S. Micali, ``Probabilistic encryption,'' \emph{J.\ Comput.\ Syst.\ Sci.}, 28(2):270--299, 1984.
\bibitem{mack2017level} C. A. Mack, T. A. Brunner, X. Chen, and L. Sun, ``Level crossing methodology applied to line-edge roughness characterization,'' \emph{Proc.\ SPIE}, 10145:101450Z, 2017.
\bibitem{smith2009} G. Smith, ``On the foundations of quantitative information flow,'' in \emph{FoSSaCS}, LNCS 5504, pp.\ 288--302, Springer, 2009.
\bibitem{dwork2006} C. Dwork, ``Differential privacy,'' in \emph{ICALP}, LNCS 4052, pp.\ 1--12, Springer, 2006.
\bibitem{cover2006} T. M. Cover and J. A. Thomas, \emph{Elements of Information Theory}, 2nd ed., Wiley, 2006.
\bibitem{standaert2009} F.-X. Standaert, T. G. Malkin, and M. Yung, ``A unified framework for the analysis of side-channel key recovery attacks,'' in \emph{EUROCRYPT 2009}, LNCS 5479, pp.\ 443--461, Springer, 2009.
\bibitem{choudary2013} O. Choudary and M. G. Kuhn, ``Efficient template attacks,'' in \emph{CARDIS 2013}, LNCS 8419, pp.\ 253--270, Springer, 2014.
\bibitem{ledoit2004} O. Ledoit and M. Wolf, ``A well-conditioned estimator for large-dimensional covariance matrices,'' \emph{J.\ Multivariate Anal.}, 88(2):365--411, 2004.
\bibitem{alvim2020} M. S. Alvim, K. Chatzikokolakis, A. McIver, C. Morgan, C. Palamidessi, and G. Smith, \emph{The Science of Quantitative Information Flow}, Springer, 2020.
\bibitem{shokri2017} R. Shokri, M. Stronati, C. Song, and V. Shmatikov, ``Membership inference attacks against machine learning models,'' in \emph{IEEE Symp.\ Security and Privacy}, pp.\ 3--18, 2017.
\bibitem{dinur2003} I. Dinur and K. Nissim, ``Revealing information while preserving privacy,'' in \emph{PODS}, pp.\ 202--210, ACM, 2003.
\end{thebibliography}
\end{document}